\documentclass[10pt,journal,compsoc]{IEEEtran}
\ifCLASSOPTIONcompsoc
  \usepackage[nocompress]{cite}
\else
  \usepackage{cite}
\fi

\ifCLASSINFOpdf
  \usepackage[pdftex]{graphicx}
\else
  \usepackage[dvips]{graphicx}
\fi

\usepackage[fleqn]{amsmath}
\usepackage{mathtools}
\usepackage{amsthm}
\usepackage{amssymb}

\usepackage{algorithmic}

\usepackage{array}

\usepackage[caption=false,font=footnotesize,labelfont=sf,textfont=sf]{subfig}

\usepackage{fixltx2e}
\usepackage{stfloats}
\usepackage{url}
\usepackage{mathrsfs}
\usepackage{bbold}
\usepackage[linesnumbered, boxed, ruled,vlined]{algorithm2e}

\usepackage[shortlabels]{enumitem}

\newtheorem{theorem}{Theorem}[section]

\newtheorem{proposition}[theorem]{Proposition}

\begin{document}

\title{Sequence Covering for Efficient Host-Based Intrusion Detection}

\author{Pierre-Francois Marteau
\IEEEcompsocitemizethanks{\IEEEcompsocthanksitem P.-F.. Marteau was with the Institut de Recherche en Informatique et Systemes Aleatoires (IRISA) Laboratory, Universite Bretagne Sud, France.\protect\\
E-mail: https://people.irisa.fr/Pierre-Francois.Marteau/
}
\thanks{Manuscript received xxx, 2017; revised yyy.}}

\markboth{Journal of \LaTeX\ Class Files,~Vol.~14, No.~X, February~2018}%
{Shell \MakeLowercase{\textit{et al.}}: Bare Demo of IEEEtran.cls for Computer Society Journals}

\IEEEtitleabstractindextext{%
\begin{abstract}
This paper introduces a new similarity measure, the covering similarity,  that we formally define for evaluating the similarity between a symbolic sequence and a set of symbolic sequences. A pair-wise similarity can also be directly derived from the covering similarity to compare two symbolic sequences.  An efficient implementation to compute the covering similarity is proposed that uses a suffix-tree data-structure, but other implementations, based on suffix-array for instance, are possible and possibly necessary for handling very large-scale problems. We have used this similarity to isolate attack sequences from normal sequences in the scope of Host-based Intrusion Detection. We have assessed the covering similarity on two well-known benchmarks in the field. In view of the results reported on these two datasets for the state of the art methods, and according to the comparative study we have carried out based on three challenging similarity measures commonly used for string processing or in bioinformatics, we show that the covering similarity is particularly relevant to address the detection of anomalies in sequences of system calls.
\end{abstract}

}

\maketitle

\IEEEdisplaynontitleabstractindextext

%
\IEEEpeerreviewmaketitle

\IEEEraisesectionheading{\section{Introduction}\label{sec:introduction}}
\IEEEPARstart{I}{ntrusion} Detection Systems (IDS) are more and more heavily challenged by intrusion \textit{scenarios} developed by today's hackers. The number of reported intrusion incidents has dramatically increased during the last few years with very serious consequences for organizations, companies and individuals. As an example, the Troyan horse TINBA (which stands for TINy BAnker) has targeted with apparent success  the worldwide banking system during the last three years \cite{regev2014, Bonderud2016, Bonderud2016b}. The detection of zero-day attacks (attacks that have never been detected before) is even more challenging since no pattern or signature characterizing this kind of attack can be used to identify it. Furthermore, with the development of the IoT, the rate of the production of sequences of system calls, i.e. sequential data used to access, manage, or administrate connected equipments, is exploding.  Hence, the need to develop and use efficient intrusion detection algorithms that can identify, isolate and handle suspicious patterns in sequential information flows is evermore pressing with time.

This paper addresses the detection of (unknown) anomalies in symbolic sequential data with a specific focus on sequences of system calls within the scope of intrusion detection. A system call is a request to the kernel of an operating system to provide a service on behalf of the user's program. System calls are used to manage the file system and available hardware, control processes, and to provide interprocess communication. Thus, a sequence of system calls corresponds to the sequential list of service requests sent by a process to the kernel, and as such, it constitutes a trace that describes the behavior of the monitored process.

The detection of abnormal system call sequences is challenging for the following reasons:
\begin{enumerate}
\item the anomalies are contextual: the occurrence of a system call can be considered as abnormal given its context of occurrence, basically given the subsequences of system calls that occur before and after it,
\item the anomalies can be collective: a subsequence as a whole can be considered as abnormal,
\item the variability of system call sequences is very high mainly because they are of varying length (a low number to a few thousand system calls) and the alphabet on which the system calls are defined is quite high (more than 300 for the Linux system),
\item if subsequences can be seen as discriminative features, the combinatorics increases dramatically with the size of the subsequences. This prohibits the use of $n$-gram features for example, when $n$ is above a low number of system calls (4, 5). 
\end{enumerate} 

Bioinformatics has proposed over the years a great number of similarity measures \cite{Smith1981, Needleman1970, Korf2003} to compare pairs of sequences or provide multiple sequence alignments for sets of sequences. These editing distances cover partially the above-mentioned difficulties and have been used with great success to evaluate similarities and dissimilarities in sequences of nucleotides, some of them allowing for the isolation of abnormal subsequences \cite{Lamperti1992}. In addition, they have been successfully adapted to cope with sequences of system calls \cite{Quan2012}. However, the quadratic time complexity of most of these similarity measures and the need for the computation of the whole similarity matrix limit their use to small to medium size problems. Approximate computations of these measures along with extremely parallel hardware are required to address large-scale problems \cite{Ho2017}.

In this context, designing an algorithm that can somehow profit from subsequences of any size to evaluate the pairwise similarity of sequences with a significantly lower time complexity is challenging. Furthermore, in the scope of anomaly-based intrusion detection, if this similarity measure can be used to compare a test sequence to a set of normal sequences, it can result in a quite and efficient intrusion detection approach.  
  
Our main contribution is the description of an efficient algorithm (SC4ID, which stands for Sequence Covering for Intrusion Detection) based on the concept of a so-called \textit{optimal-covering} of a sequence by a series of  subsequences extracted from a predefined set of sequences. We demonstrate in this paper that this algorithm is efficient, at least in terms of accuracy and response time to isolate attack sequences that have never been observed (zero-day settings). Indeed, as its implementation is based on generalized suffix-trees (or suffix-arrays), it requires a relatively large memory overhead.

The second section of this paper briefly reports the related key works. We detail the SC4ID algorithm in the third section and evaluate it in the fourth section on two distinct sequences of system call benchmarks. On the smallest size benchmark we compare the proposed similarity measure with three baseline similarity measures commonly used in bioinformatics to estimate the similarity of symbolic sequences. We discuss our results in the fifth section before concluding this study in section six.

\section{Context and Related works}

In a broad sense we address anomaly detection in sequential data \cite{Chandola2009} while focusing on intrusion detection in cyber-physical systems. Intrusion \cite{Scarfone2007} refers to possible security breaches in (cyber-)systems, namely malicious activity or policy violations. It covers both intrusions \textit{per se}, i.e. attacks from the outside, and misuse, i.e. attacks from within the system. An intrusion detection system (IDS) is thus a device that monitors a system for detecting potential intrusions. The IDS will be referred to as  NIDS if the detection takes place on a network and HIDS if it takes place on a host of a network. Furthermore, we distinguish i) signature-based IDS approaches, that detects attacks by looking for predefined specific patterns, such as byte sequences in network packets, or known malicious sequences of instructions  used by malware, to ii) anomaly-based intrusion detection systems that were primarily introduced to detect unknown attacks (zero-day attacks).

In this work we exclusively address host intrusion detection system (HIDS) through semi-supervised anomaly-based approaches. Since Forrest's pioneering  work \cite{Forrest1996} most of HIDS (at least in the UNIX/LINUX sphere) use system call sequences
as their primary source of information. 
Generally, sequences of system calls  are represented as sequences of integer symbols, for which the order of occurrence of the symbol is of crucial importance. Numerous work and surveys have been published in the area of anomaly detection in the scope of intrusion detection, see \cite{Liao2013} \cite{Bhuyan2014}, \cite{Hodo2017} for recent studies. If we reduce the area of interest to anomaly detection in sequential data, four avenues for handling symbolic sequences are mainly followed:\\

\textbf{Window-based approaches} \cite{Hofmeyr1998} are quite popular since a fixed size window enables a wide range of statistical, knowledge-based and machine learning techniques to be applied in a straightforward manner. A fixed-size window is first defined, and then it progressively slides along the tested sequence. Each window (basically a fixed-size subsequence) is in general represented by a feature vector. Then, models such as the one class Support Vector Machine \cite{Wang2004}, Multi Layer Perceptron and Convolutional Neural Networks \cite{Yao2006} are used to provide a score for deciding whether an anomaly is present or not.

In this framework, an aggregation of each  window (that is sliding all along the sequence) score is necessary to get a global decision at the sequence level.  The prediction score is used to detect and locate the position of the anomalies inside the test sequence if any.\\


\textbf{Global kernel-based  approaches} \cite{Budalakoti2009} process each sequence as a whole and a pair-wise sequence kernel (string kernel) is used to provide the sequence space with a similarity measure. The $k$-Near-Neighbor rule or any of the so-called kernel machine methods can then be applied to model the 'normal' clusters and isolate the 'anomalies'. These approaches find their roots in text processing \cite{Levenshtein66} (Levenshtein's distance) or in Bioinformatics \cite{Smith1981, Coull2003} (Smith and Watermans),  \cite{Needleman1970} (Needleman-Wunsch), \cite{Kundu2010} (BLAST) and Longest Common Subsequence (LCS) or Longest Similar Subsequence \cite{Wan2006} \cite{Sureka2015}. Such methods do not seem to outperform window-based approaches and are in general much more costly in  term of algorithmic complexity than state of the art methods. \\

\textbf{Generative approaches}, essentially  Hidden Markov Models (HMM) \cite{Qiao2002} \cite{Hu2010} \cite{Jain2012}, Conditional random Fields (CRF) \cite{Gupta2007} \cite{Taub2013} or Recurrent Neural Networks (RNN, LSTM, etc.) \cite{Sheikhan2012} \cite{Kim2016} have been used with apparent success on various intrusion detection tasks, such as payload analysis or Network Layer Intrusion Detection or HIDS. However, the choice of parameters such as the order of the Markovian dependency, number of hidden variables, etc., is often the result of a compromise to avoid over-fitting, and long-term dependency is not necessarily easily modeled. \\

\textbf{Language-based approaches} have been proposed initially to extract very simple n-gram features to enhance a vector space model similar to the one used in text mining \cite{Rieck2007}, \cite{Orisaniya2015}. Recently,  a much ambitious model has been proposed that proposes to enact phrases and sentences, hence a language, from sequences of system calls ~\cite{Creech2014}. Nevertheless, these approaches suffer from the combinatorics explosion. When simple $n$-grams models are used (with $n$ lower than $5$ or $7$) the size of the vector space model is very high (several millions of dimension) and in general the lack of available data to train the model limits its accuracy. In the case of  Creech et al. approach \cite{Creech2014}, the combinatorics is much higher with an estimated feature space dimension of $10^{16}$ which makes this model intractable for common hardware.\\

The approach we present below relates to the (global) kernel-based family. Each sequence is thus considered as a whole, regardless of its length. All the specificity and novelty of the method relies on a quite simple similarity measure, that we call covering similarity. To our knowledge, this similarity measure has not been yet proposed for sequence comparison, specifically in the context of intrusion detection.  It is defined to evaluate how close a sequence is to a set of supposedly 'normal' sequences. A simple threshold is  used to decide whether an unknown sequence is 'normal' or should be considered as an anomaly. The main advantages of our approach are:
\begin{itemize}
\item apart from a decision threshold, it is parameter-free, in particular it does not rely on a window size,
\item it is incremental and can be set up 'online',
\item it is interpretable since it easily enables the location of abnormal areas in long sequences,
\item it supports a very efficient instance selection scheme to iteratively improve the 'normal' model, without overloading it with unnecessary instances,
\item it is quite efficient compared to other machine learning based models and scales well compared to classical sequence alignment kernels (which, in general, are  at least in $O(n^2)$ complexity) such as string kernels, the Longest Common Subsequence or the Smith and Waterman similarities, due to the suffix-tree and suffix-array data structures on which the implementation of our algorithm relies. 
\item Furthermore, it does not rely on pairwise distance calculation between sequences, but on the distance between a sequence and a set of reference sequences, which also drastically reduces its computational complexity.
\end{itemize}

\section{The SC4ID algorithm}
The overall principle of the SC4ID algorithm is straightforward. It is depicted in algorithm \ref{fig:SC4IDprinciple} and presented in Fig. \ref{alg:SC4ID}. Given a set of sequences considered as normal sequences, $S$, and a threshold $\sigma \in [0;1]$, SC4ID evaluates the similarity of an unknown sequence $s$ with the elements of $S$ according to a similarity measure $\mathscr{S}(s,S)$ that consists in optimally covering $s$ with subsequences of elements of $S$. SC4ID can be considered as an implicit semantic-based approach if we interpret the covering subsequences as phrases constructed from $S$, the corpus of training sequences. SC4ID also relates to editing distances or time elastic measures such as the Levenshtein's distance. But instead of addressing the matching problem through local editing operations (substitution, insertion, deletion), it evaluates the minimal number of subsequences that are required to build a complete covering of a given sequence. Furthermore, it allows for subsequences to be swapped, an operation that is very costly to implement in the editing distance framework. Most importantly, SC4ID  does not evaluate pair-wise sequence similarities to construct a whole similarity matrix as for the Levenshtein's distance, but directly computes the similarity between a sequence and a set of sequences which drastically reduces its algorithmic complexity. 

Behind this covering principle, we potentially manage to use any subsequence that can be drawn from the set $S$, without constraining neither the length of the subsequences nor their number. The covering is in itself informative and enables the location of anomalies spread inside a long sequence. Furthermore one can interpret the covering as a set of words that describes the sequence, seen as a phrase (as conceptually proposed in \cite{Creech2013}). Indeed the fact that a fast algorithm exists to evaluate S-optimal coverings makes it very relevant in the context of anomaly detection in sequential data.

Finally, as described in algorithm \ref{alg:SC4ID}, if the covering similarity is above the threshold $\sigma$ then the sequence $s$ is considered as normal, otherwise it will be considered as an anomaly. 

The only parameter in the SC4ID algorithm is $\sigma$. In a semi-supervised anomaly-based detection framework applied to intrusion detection, all methods provide a confidence measure (e.g a probability, a distance, a similarity measure, etc.) and the definition and setting up of an $\sigma$-like parameter in order to make the final decision. We discuss the tuning of this parameter for the SC4ID algorithm in subsection \ref{subsec:sigma}.

\begin{figure}
\center\includegraphics[scale=.5]{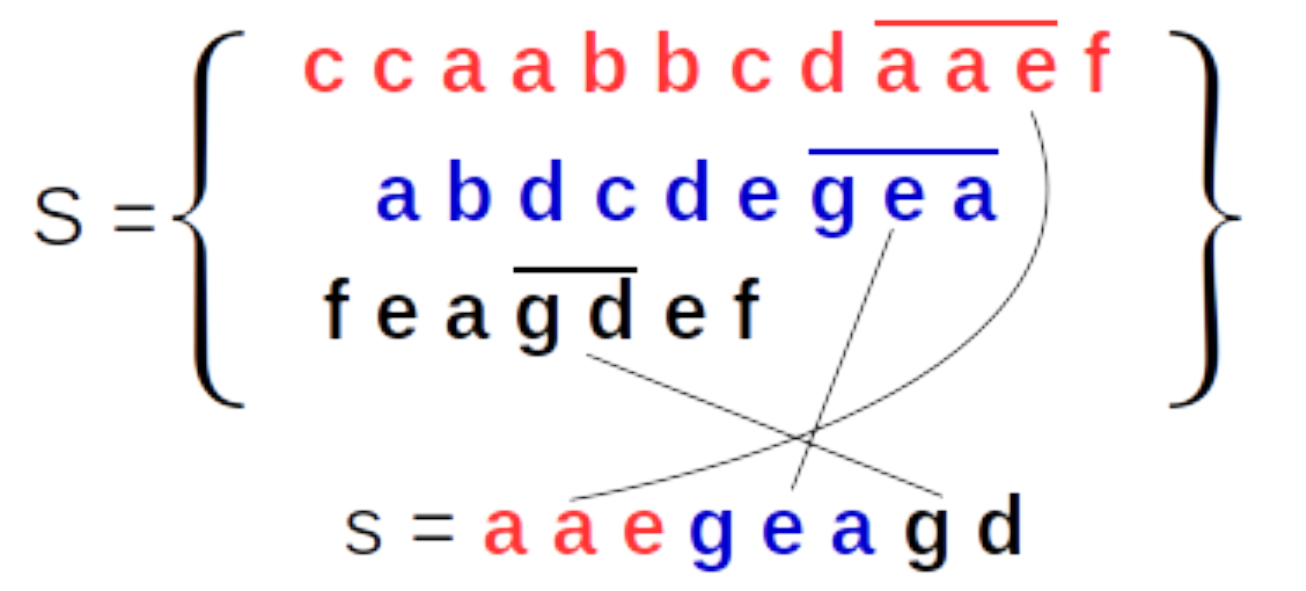}
\caption{Principle of SC4ID. In the example, subsequence $s$ is optimally covered by three subsequences drawn from sequences in $S$, $[aae, gea,gd]$, which is a $S$-optimal covering for $s$. }
\label{fig:SC4IDprinciple}
\end{figure}

The specificity and novelty of the algorithm lie in the way similarity $\mathscr{S}(s,S)$ is defined. 
We introduce hereinafter some definitions and notation to detail the formal definition of this similarity measure.

\subsection{Definitions and notation}

Let $\Sigma$ be a finite alphabet and let $\Sigma^*$ be the set of all sequences (or strings) defined over $\Sigma$. We note $\epsilon$ as the empty sequence.

Let $S \subset \Sigma^*$ be any set of sequences, and let $S_{sub}$ be the set of all subsequences that can be extracted from any element of $S\cup \Sigma$. We denote by $\mathscr{M}(S_{sub})$ the set of all the multisets\footnote{A multiset is a collection of elements in which elements are allowed to repeat; it may contain a finite number of indistinguishable copies of a particular element.} that we can compose from the elements of $S_{sub}$.

$c \in \mathscr{M}(S_{sub})$ is called a partial covering of sequence $s \in \Sigma^*$ if and only if 
\begin{enumerate}
\item all the subsequences of $c$ are also subsequences of $s$,
\item indistinguishable copies of a particular element in $c$ correspond to distinct occurrences of the same subsequence in $s$.
\end{enumerate}

If $c \in \mathscr{M}(S_{sub})$ entirely covers $s$, meaning that we can find an arrangement of the elements of $c$ that covers $s$ entirely, then we will call it a full covering for $s$.


Finally, we call a $S$-optimal covering of $s$ any full covering of $s$ which is composed with a minimal number of subsequences in $S_{sub}$.

Let $c^*_S(s)$ be a $S$-optimal covering of  $s$.

We define the covering similarity measure between any non-empty sequence $s$ and any set $S \subset \Sigma^*$ as
\begin{equation}
\mathscr{S}(s,S) = \frac{|s|-|c^*_S(s)|+1}{|s|} 
\label{eq:coveringSimilarity}
\end{equation}
where $|c^*_S(s)|$ is the number of subsequences composing a $S$-optimal covering of $s$, and $|s|$ is the length of sequence $s$.\\

Note that in general $c^*_S(s)$  is not unique, but since all such coverings have the same cardinality, $|c^*_S(s)|$, $\mathscr{S}(s,S)$ is well defined. 


Properties of $\mathscr{S}(s,S)$:
\begin{enumerate}
\item If $s$ is a non-empty subsequence in $S_{sub}$, then  $\mathscr{S}(s,S) = 1$ is maximal.
\item In the worst case, the $S$-optimal covering of $s$ has cardinality equal to $|s|$, meaning that it is composed only with subsequences of length $1$. In that case, $\mathscr{S}(s,S) = \frac{1}{|s|}$ is minimal.
\item If $s$ is non-empty, $\mathscr{S}(s,\emptyset) = \frac{1}{|s|}$ (note that if $S=\emptyset$, $S_{sub}=\Sigma$).
\end{enumerate} 

Furthermore, as $\epsilon$ is a subsequence of any sequence in $\Sigma^*$, we define, for any set $S \subset \Sigma^*$, $\mathscr{S}(\epsilon,S) = 1.0$ 

Note that the covering similarity between a sequence and a set of sequences as defined in Eq. \ref{eq:coveringSimilarity} enables the definition of a covering similarity measure on the sequence set itself. For any pair of sequences $s_1$, $s_2$ this measure is defined as follows: 
\begin{equation}
\mathscr{S}_{seq}(s_1,s_2) = \frac{1}{2} (\mathscr{S}(s_1,\{s_2\})+ \mathscr{S}(s_2,\{s_1\}))  
\label{eq:SeqCoveringSimilarity}
\end{equation}

where $\mathscr{S}$ is defined in Eq. \ref{eq:coveringSimilarity}.

\begin{algorithm}
\SetKwData{Left}{left}\SetKwData{This}{this}\SetKwData{Up}{up}
\SetKwFunction{Union}{Union}\SetKwFunction{FindCompress}{FindCompress}
\SetKwInOut{Input}{input}\SetKwInOut{Output}{output}
\Input{$S \subset \Sigma^*$, a set of sequences}
\Input{$s \in \Sigma^*$, a test sequence }
\Input{$\sigma \in [0,1]$, a threshold value }
\Output{a decision value: 'normal' or 'anomaly'}
\BlankLine
Provide a $S$-optimal covering of $s$\;
Evaluates $\mathscr{S}(s,S)$ according to Eq. \ref{eq:coveringSimilarity}\;
\lIf{$\mathscr{S}(s,S) \ge \sigma$}{\Return ['normal', $\mathscr{S}(s,S)$ ]}
\lElse{\Return ['anomaly', $\mathscr{S}(s,S)$ ]}
\BlankLine
\caption{SC4ID\label{alg:SC4ID}}
\end{algorithm}

As an example, let us consider the following case:

\begin{tabular}{ll} 
&$s_1$ = [0,0,0,0,1,1,1,1,0,0,0,0,1,1,1,1] \\
&$s_2$ = [0,0,0,0,0,0,0,0,1,1,1,1,1,1,1,1]\\
&$S= \{s_1, s_2\}$\\
&$s_3$ = [0,0,1,1,0,0,1,1,0,0,1,1,0,0,1,1]\\
&$s_4$ = [0,1,0,1,0,1,0,1,0,1,0,1,0,1,0,1]\\
\end{tabular} 

The $S$-optimal covering of $s_3$ \footnote{([0,0,1,1][0,0,1,1],[0,0,1,1][0,0,1,1]) is a $S$-optimal covering of $s_3$} is  size $4$, hence $\mathscr{S}(s_3,S) = \frac{16-4+1}{16}=13/16$, and the $S$-optimal covering of $s_4$ \footnote{([0,1],[0,1],[0,1],[0,1],[0,1],[0,1],[0,1],[0,1]) is a  $S$-optimal covering for $s_4$} is size $8$, leading to $\mathscr{S}(s_4,S) = \frac{16-8+1}{16}=9/16$.\\

The main challenge for the SC4ID algorithm is to evaluate $\mathscr{S}(s,S)$ efficiently for a sufficiently  large $S$ and relatively long sequences $s$ such as to be able to process common sequences of system calls. This essentially requires an efficient way to get $S$-optimal coverings for tuples $(s,S)$ constructed from general sequences of system calls.

\subsection{Finding a $S$-optimal covering for any tuple $(s,S)$}

The brute-force approach to find a $S$-optimal covering for a sequence $s$ is presented in algorithm \ref{alg:Find-S-optimal}. It is an incremental algorithm that, first, finds the longest subsequence of $s$ that is contained in $S_{sub}$ and that starts at the beginning of $s$. This first subsequence is the first element of the $S$-optimal covering. Then, it searches for the next longest subsequence that is in $S_{sub}$ and which starts at the end of the first element of the covering, then adds it to the covering under construction, and it is then iterated until reaching the end of sequence $s$.  

\begin{algorithm}
\SetKwData{Left}{left}\SetKwData{This}{this}\SetKwData{Up}{up}
\SetKwFunction{Union}{Union}\SetKwFunction{FindCompress}{FindCompress}
\SetKwInOut{Input}{input}\SetKwInOut{Output}{output}
\Input{$S \subset \Sigma^*$, a set of sequences}
\Input{$s \in \Sigma^*$, a test sequence }
\Output{$c$, a ($S$-optimal) covering for $s$}
\BlankLine
$continue \longleftarrow True$\;
$start \longleftarrow 0$\;
$c^* \longleftarrow \emptyset$\;
\While{continue}{
  $end \longleftarrow start + 1$\;
  \While{$end<|s|$ and $s[start:end] \in S_{sub}$}{
  $end \longleftarrow end + 1$\;
  }
  $c \longleftarrow c^* \cup \{s[start:end-1]\}$\;
  \lIf{$end = |s|$}{$continue \longleftarrow False$}
  $start \longleftarrow end $\;
}
\Return $c$\;
\BlankLine
\caption{Find a $S$-optimal covering for $s$\label{alg:Find-S-optimal}}
\end{algorithm}

\begin{proposition} Algorithm \ref{alg:Find-S-optimal} outputs a $S$-optimal covering for sequence $s$.
\end{proposition}

\begin{proof} i) First we observe that since all the subsequences of length $1$ constructed on $\Sigma$ are included into $S_{sub}$, algorithm \ref{alg:Find-S-optimal}, by construction, outputs a full covering of $s$ (meaning that $s$ is entirely covered by the subsequences of the covering provided by the algorithm).  

ii) Second we observe that, for all $s_1$ and $s_2$ in $\Sigma^*$ such that $s_1$ is a subsequence of $s_2$, and any $S \subset \Sigma^*$, $|c^*_S(s_1)| \le |c^*_S(s_2)|$.

We finalize the proof by induction on $n$, the cardinality (the size) of the coverings. 

The proposition is obviously true for $n=1$: for all sequence $s$ for which a covering of size $1$ exists (meaning that $s$ is a subsequence of one of the sequences  in $S$), algorithm \ref{alg:Find-S-optimal} finds the S-optimal covering that consists of $s$ itself.

Then, assuming that the proposition holds for  $n$, such that $n \ge 1$ (IH),  we consider a sequence $s$ that admits a $S$-optimal covering of size $n+1$.

Let $s=s_1+\overline{s}_1$, be the decomposition of $s$ according to the full covering provided by algorithm \ref{alg:Find-S-optimal}, where $s_1$ is the prefix of the covering (first element) and $\overline{s}_1$ the remaining suffix subsequence (concatenation of the remaining covering elements). $+$ is the sequence concatenation operator. Similarly, Let $s=s^*_1+\overline{s}^*_1$, be the decomposition of $s$ according to a $S$-optimal covering of $s$. Essentially, $s^*_1$, which is also a prefix of $s$, is a subsequence of $s_1$ (otherwise, since $s^*_1$ is in $S_{sub}$, algorithm \ref{alg:Find-S-optimal} would have increased the length of $s_1$ at least to the length of $s^*_1$). Hence, $\overline{s}_1$ is a subsequence of $\overline{s}^*_1$ and, according to ii), $|c^*_S(\overline{s}_1)| \le |c^*_S(\overline{s}^*_1)|=n$. This shows that $\overline{s}_1$ is a sequence that admits a $S$-optimal covering, $c^*_S(\overline{s}_1)$, of a size at most equal to $n$. According to (HI), algorithm \ref{alg:Find-S-optimal} returns such an optimal covering for  $\overline{s}_1$. This shows that the covering $\{s_1\} \cup c^*_S(\overline{s}_1)$ that is returned by algorithm \ref{alg:Find-S-optimal} for the full sequence $s$, is at most size $n+1$, meaning that it is actually a  $S$-optimal covering for $s$ of size $n+1$. Hence, by induction, the proposition is true for all $n$, which proves the proposition.
\end{proof}

\SetKwProg{Fn}{Function}{}{}

\begin{algorithm}
\Fn{breakDichoSearch($s$, $t_b$, $t_e$, $S$)}{
\SetKwData{Left}{left}\SetKwData{This}{this}\SetKwData{Up}{up}
\SetKwFunction{Union}{Union}\SetKwFunction{FindCompress}{FindCompress}
\SetKwInOut{Input}{input}\SetKwInOut{Output}{output}
\Input{$s \in \Sigma^*$, a test sequence }
\Input{$t_b < t_e < |s|$, the index segment in which looking for the break}
\Input{$S \subset \Sigma^*$, a set of sequences}
\Output{$t$, the searched breaking index position}
\BlankLine
$t \longleftarrow \lfloor(t_b + t_e)/2 \rfloor$\;
\eIf{$t=t_b$ and $s[t_b:t_e] \in S_{sub}$}{\Return t+1}{\Return t}
\eIf{$s[t_b:t] \in S_{sub}$}{\Return breakDichoSearch($s$, $t$, $t_e$, $S$)\;}
	{\Return breakDichoSearch($s$, $t_b$, $t$, $S$)\;}
\BlankLine
\caption{Find the first break location in $s$ between positions $t_b$ and $t_e$\label{alg:breakDichotomicSearch}}}
\end{algorithm}

\begin{algorithm}
\SetKwData{Left}{left}\SetKwData{This}{this}\SetKwData{Up}{up}
\SetKwFunction{Union}{Union}\SetKwFunction{FindCompress}{FindCompress}
\SetKwInOut{Input}{input}\SetKwInOut{Output}{output}
\Input{$S \subset \Sigma^*$, a set of sequences}
\Input{$s \in \Sigma^*$, a test sequence }
\Output{$c^*$, a $S$-optimal covering for $s$}
\BlankLine
$continue \longleftarrow True$\;
$start \longleftarrow 0$\;
$c^* \longleftarrow \emptyset$\;
\While{continue}{
  $t \longleftarrow$ breakDichoSearch($s$, start, $|s|$, $S_{sub}$)\;
  $c^* \longleftarrow c^* \cup \{s[start:t-1]\}$\;
  \lIf{$t = |s|$}{$continue \longleftarrow False$}
  $start \longleftarrow t $\;
}
\Return $c^*$\;
\BlankLine
\caption{Find using a binary search a $S$-optimal covering for $s$\label{alg:FindDicho-S-optimal}}
\end{algorithm}

\subsection{Algorithmic complexity and implementation considerations}

Algorithm \ref{alg:Find-S-optimal} requires a fast way to test the existence of a subsequence in the sequences of $S$. Similarly to many algorithms in the field of information indexing and retrieval, these algorithms face two difficulties when the size of the data (the size of $S$ and the average length of the sequences) increases, namely memory consumption and response time. These two aspects cannot be solved simultaneously and require a compromise to be found. We discuss below some potential implementations based on hashtable, suffix-tree or suffix-array.

\subsubsection{Hashtable implementation for the search of a subsequence}
If response time is ideal, and memory space is not a problem, then one can implement a hashtable that stores all the subsequences of $S$ (more precisely the elements of $S_{sub}$). By doing so, we will be able to know if a subsequence is a member of $S_{sub}$ in $O(1)$ time complexity. Hence, the time complexity to find the $S$-optimal covering for a sequence $s$ of average size $n$ will be $O(n)$ for this implementation. 

On the other hand, if the average length of the sequences in $S$ is $n$, then, the space required to store all the elements of $S_{sub}$ is expressed in $O(n^2\cdot |S|)$. 

This could be feasible for small size problems. However, for long sequences, e.g. with an average length of $10^5$ elements, and large $S$, e.g. $10^6$ sequences, we need to store $O(10^{16})$ subsequences in the hashtable which is not feasible in practice on common hardware.

\subsubsection{Suffix-tree and suffix-array implementations for the search of a subsequence}
For medium to large size problems, we need to drastically limit the space consumption. 

In comparison to a hashtable implementation, a generalized suffix-tree  implementation \cite{Bieganski1994} \cite{Ukkonen1995} would reduce the memory requirement to $O(n \cdot |S|)$, although, in general,  with a large proportionality constant (typically $10$ to $100$ in practice), and a slight increase of the computational complexity for searching for a subsequence ($O(m)$, where $m$ is the length of the subsequence that is searched). 

In comparison, suffix-arrays \cite{Manber1990}, and specifically its enhanced implementation \cite{Abouelhoda2004}, is a space-efficient data-structure that reduces the memory consumption without losing (too much) on the response time. As they are cache friendly, a suffix-array can in practice enable handling much larger sequence sets than a suffix-tree and is much easier to parallelize. A suffix-array provides the search for a subsequence with $O(m + log(n\cdot |S|))$ average time complexity (where $n$ is the average length of the sequence in $S$).

\subsubsection{Overall time complexity}
In our current implementation, we gave priority to speed rather than memory consumption, while trying to run our algorithm on common hardware. Thus we adopted a generalized suffix-tree implementation in Python\footnote{Python implementation of Suffix Trees and Generalized Suffix Trees, https://github.com/ptrus/suffix-trees}  that we have modified to cope with sets of sequences of integers (each integer corresponding to a system call) instead of sets of strings. The sequence datasets, described below, that we have used for our experiment fit easily in a generalized suffix-tree. Hence, using a generalized suffix-tree implementation, the search for a subsequence of size $n$ is $O(n)$.

The main computing effort for algorithm \ref{alg:Find-S-optimal} is located in the second part of the test (at line 6), which consists of checking whether the subsequence $s[start:end]$ belongs to the set of subsequences $S_{sub}$ associated to $S$. As we have to iterate along the sequence $s$ to successively find the elements of its covering, algorithm \ref{alg:Find-S-optimal} would require searching $O(|s|)$ subsequences, leading to an $O(|s|^2)$ upper bound for the total time complexity. 

Indeed, a straightforward improvement of the brute-force algorithm can be achieved.  Instead of iterating along the sequence $s$ to find successively the elements of its covering,  this improvement implements a dichotomic (or binary) search to locate the extremities, that we call breaks, of the subsequences that compose the covering. This improvement is described in algorithms \ref{alg:breakDichotomicSearch} and \ref{alg:FindDicho-S-optimal}. 

The improved search algorithm has a computational complexity that is upper bounded by $O(k\cdot |s| \cdot log(|s|))$, where $k=|c^*_S({s})|$ is the size of a $S$-optimal covering for $s$. Note that, if the size, $k$, of the covering is of the same order of magnitude as the length of $s$, then the 'improved' algorithm would, in fact, require more time than the brute-force one. Hence the improvement is only achieved when the size of the covering is significantly smaller than the length of the covered subsequence, which is the case in general, except for \textit{abnormal} sequences that we are aiming to isolate. Such sequences are expected to be rare, and, in principle, we could accept the extra computing cost if we effectively manage to separate them from the flow of numerous \textit{normal} sequences that need to be processed.

In addition the previous time complexities for Algorithms \ref{alg:Find-S-optimal} and \ref{alg:SC4ID} do not depend on $|S|$, which means that increasing the size of $S$ will not impact on the processing time. This property is particularly important for applications for which $|S|$ is potentially large such as in anomaly-based intrusion detection for instance.


\begin{figure*}
\begin{tabular}{|ccc|}
      \hline
      \includegraphics[width=60mm, height=60mm]{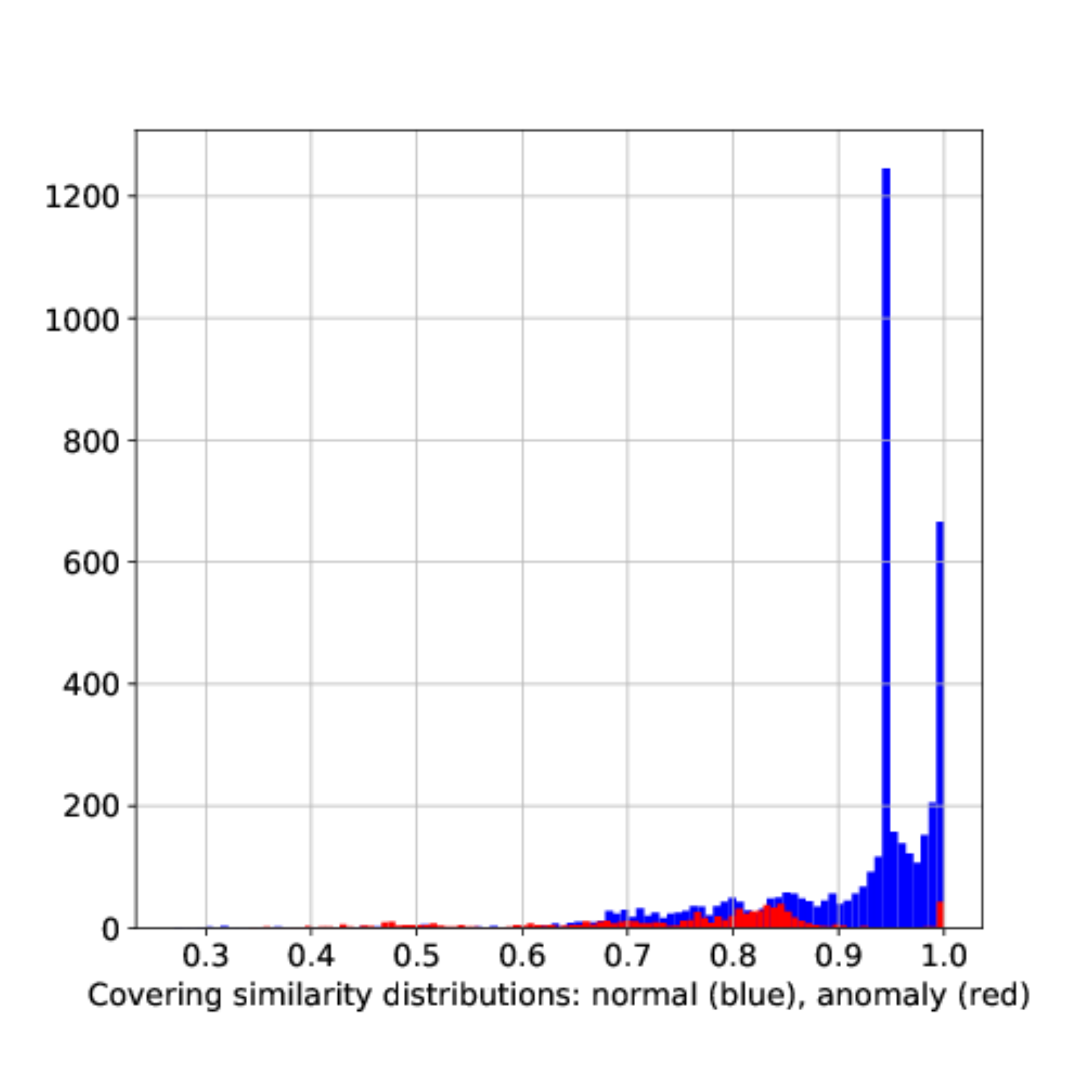} &
      \includegraphics[width=60mm, height=60mm]{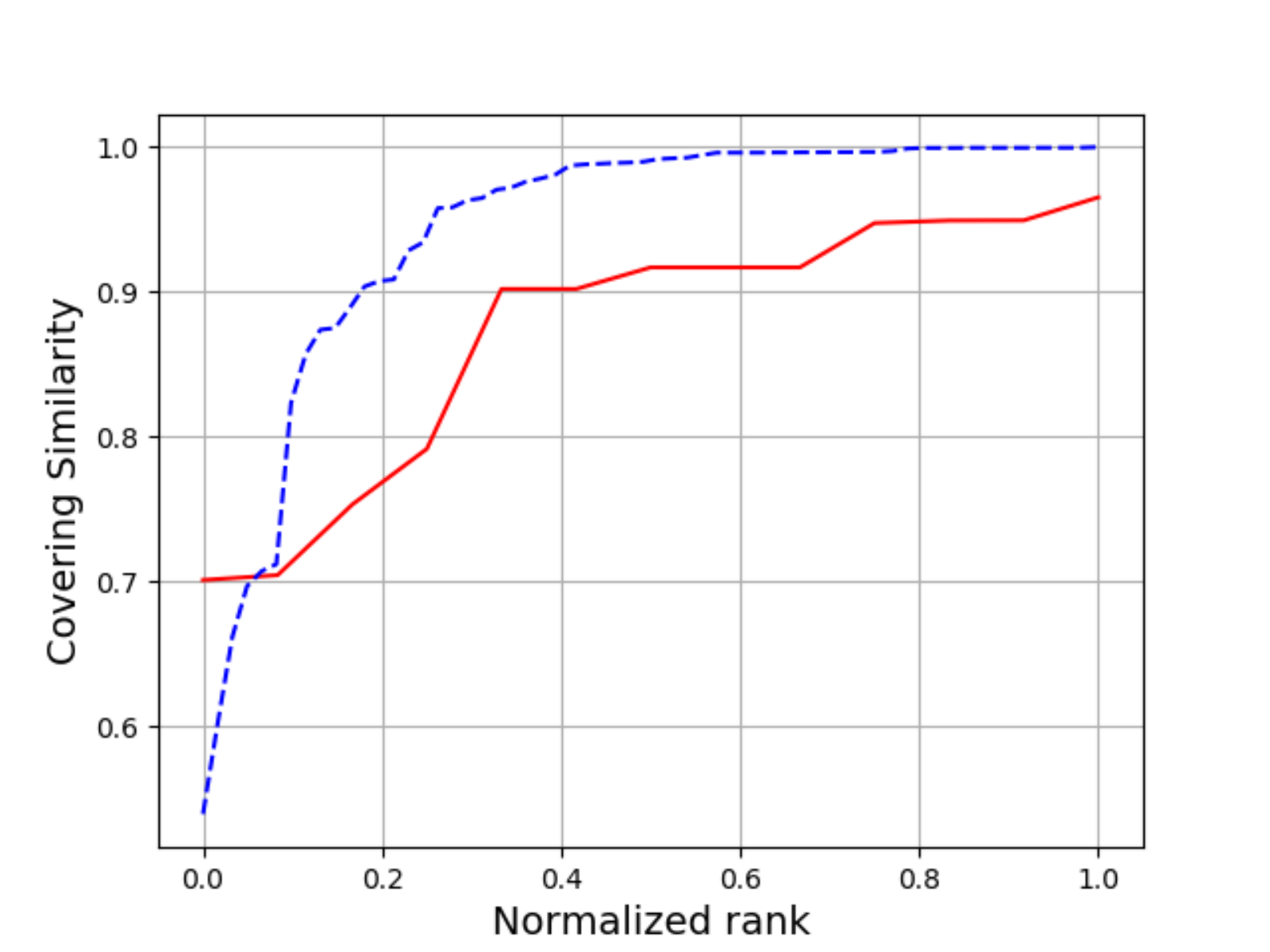} &
      \includegraphics[width=60mm, height=60mm]{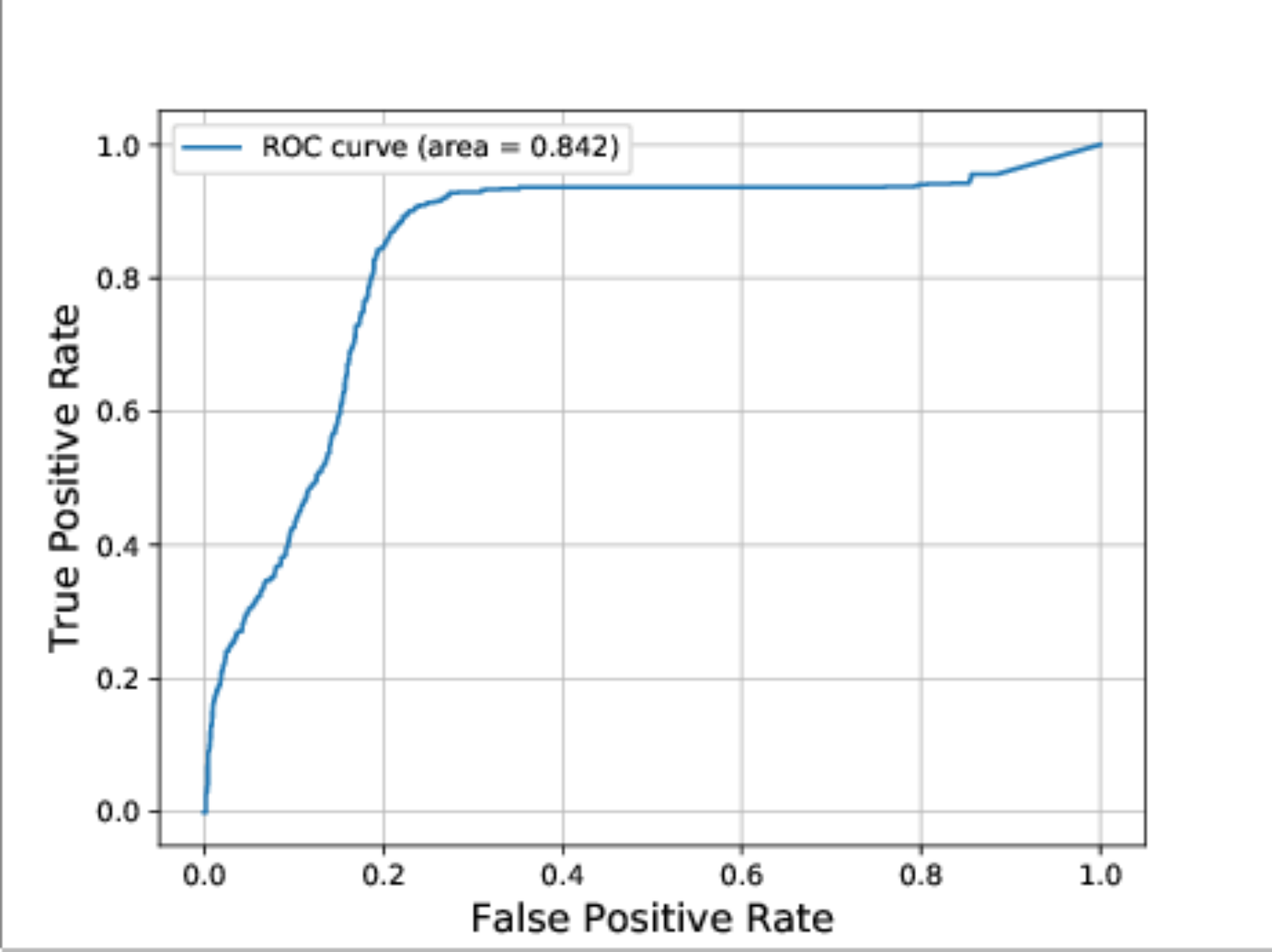} \\
      \hline
      \includegraphics[width=60mm, height=60mm]{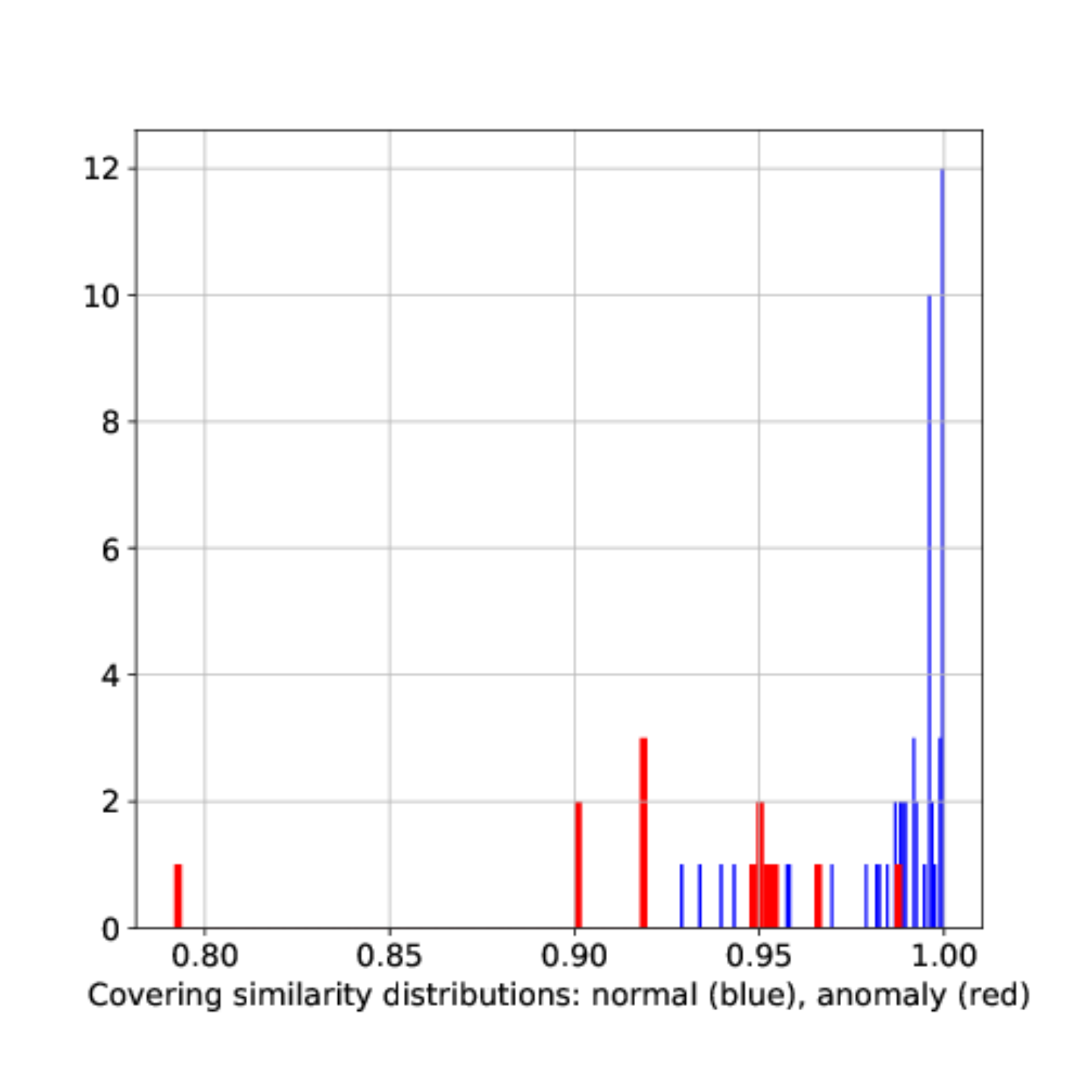} &
      \includegraphics[width=60mm, height=60mm]{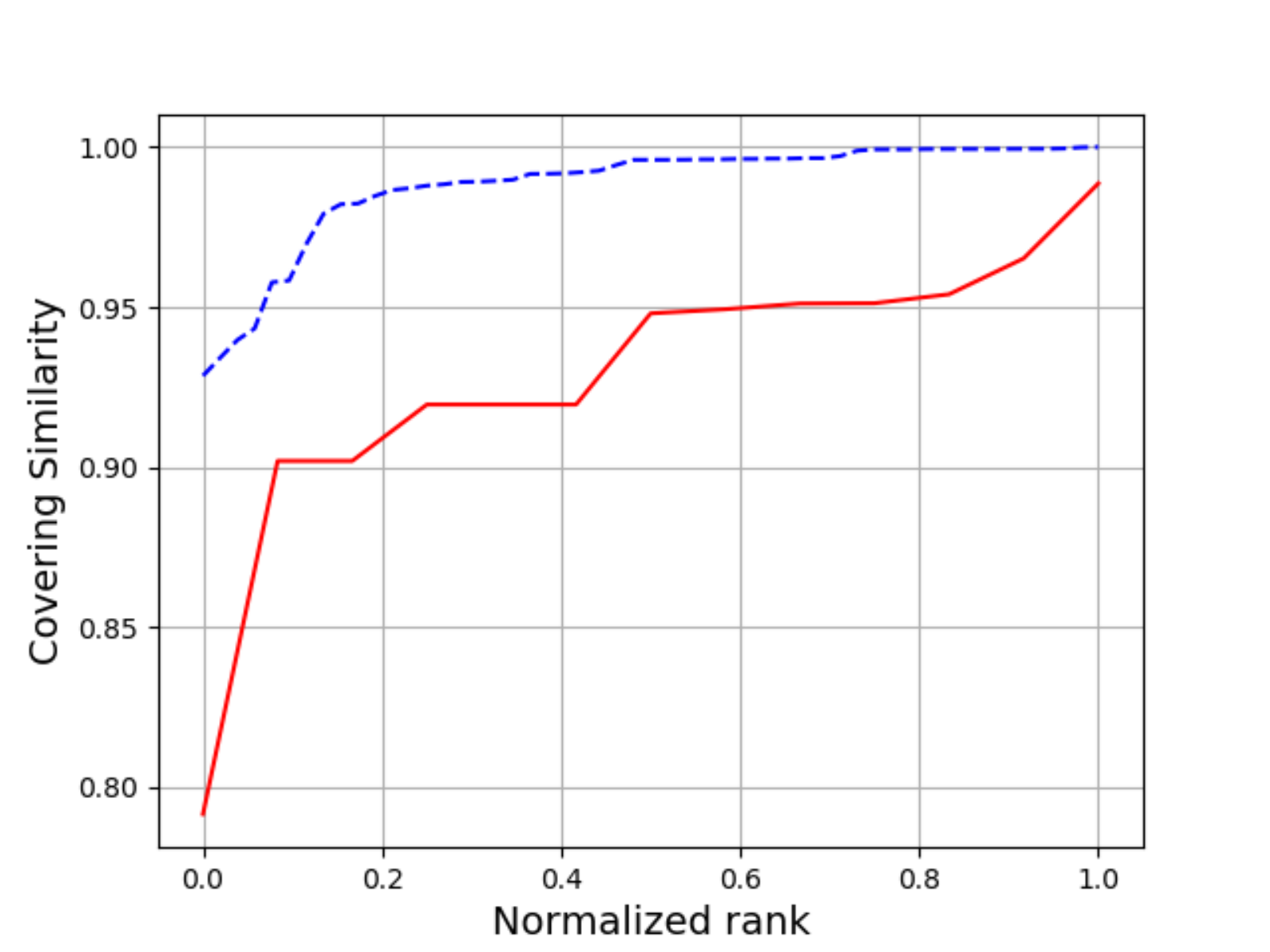} &
      \includegraphics[width=60mm, height=60mm]{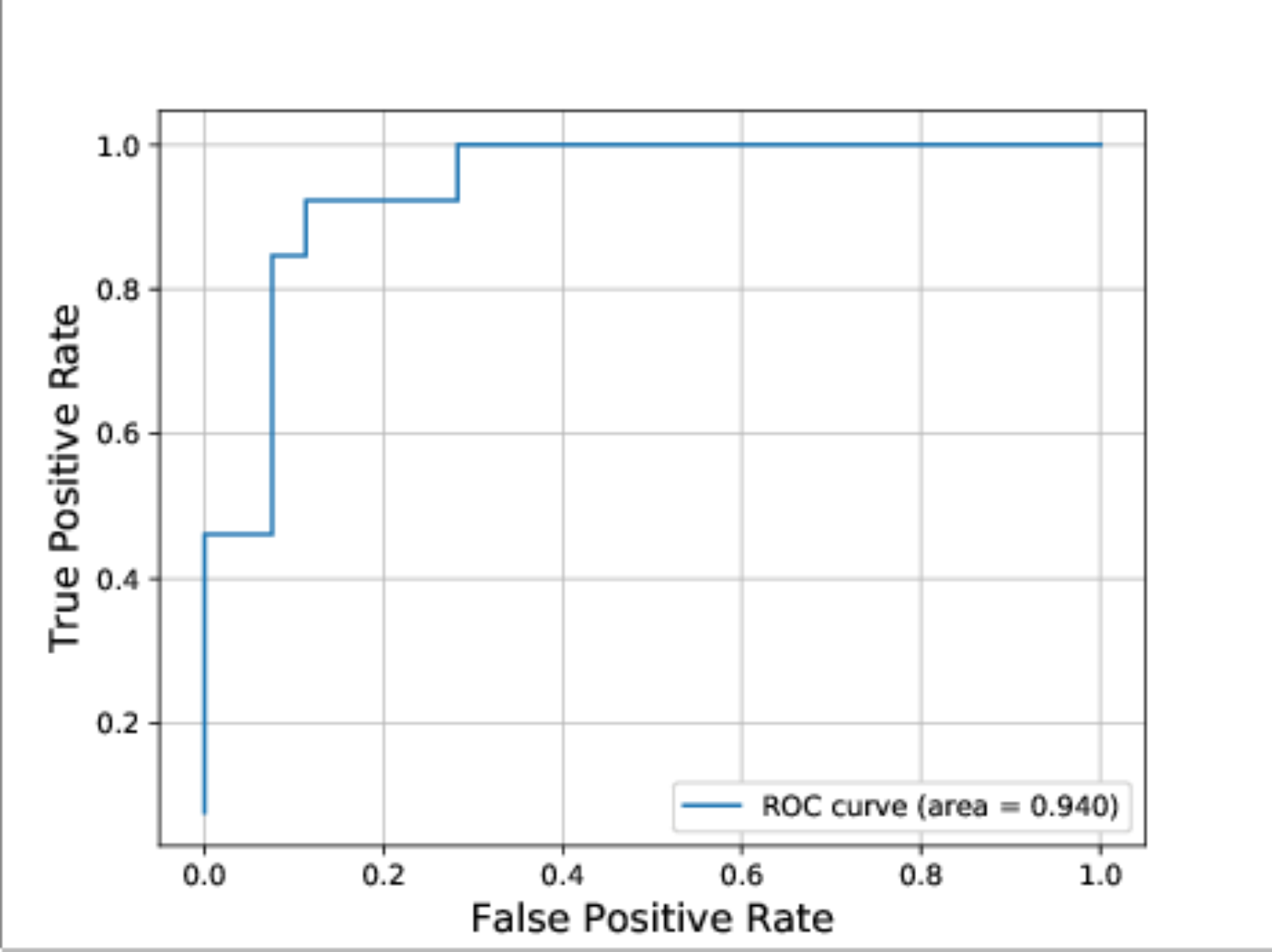} \\
      \hline
      \includegraphics[width=60mm, height=60mm]{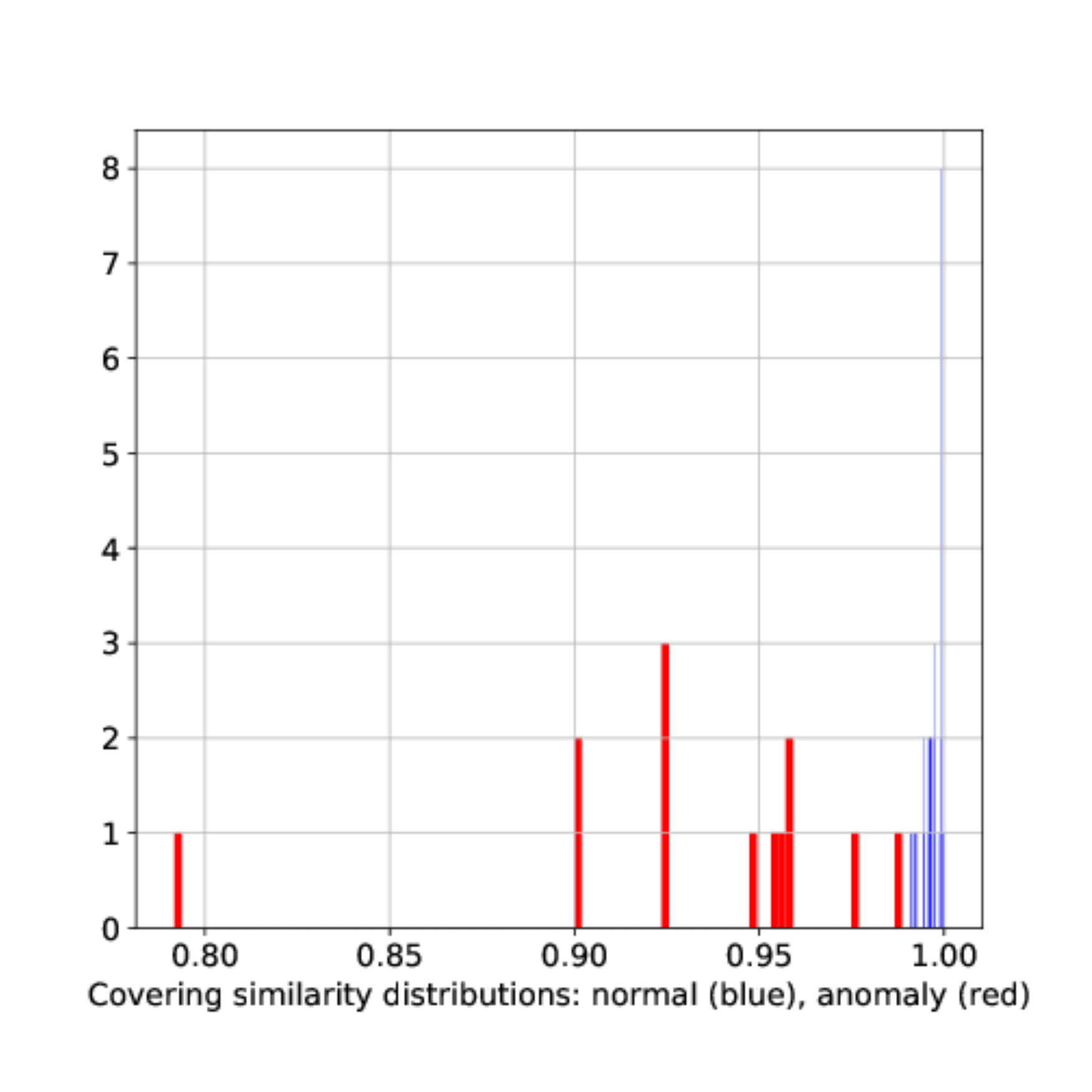} &
      \includegraphics[width=60mm, height=60mm]{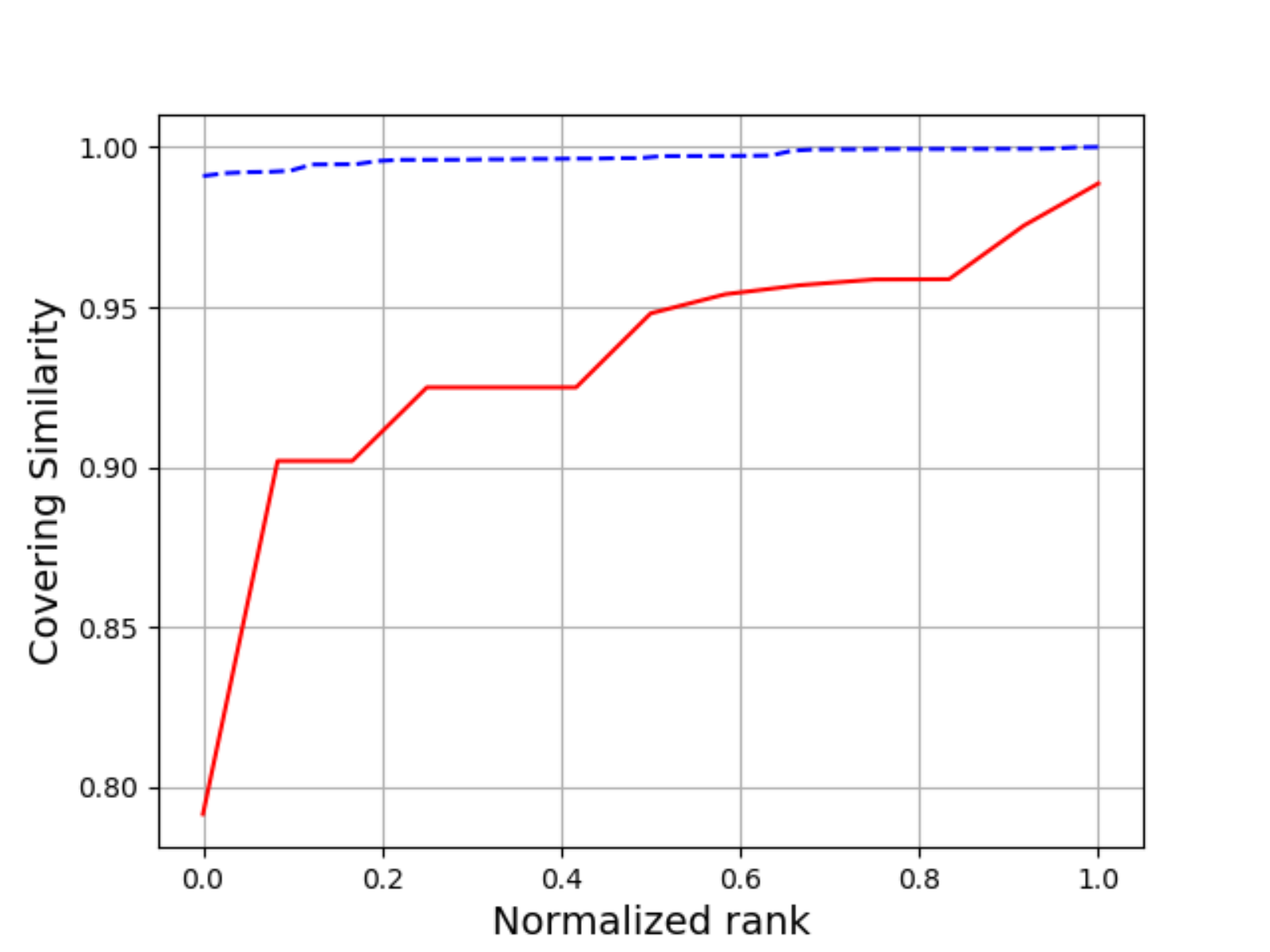} &
      \includegraphics[width=60mm, height=60mm]{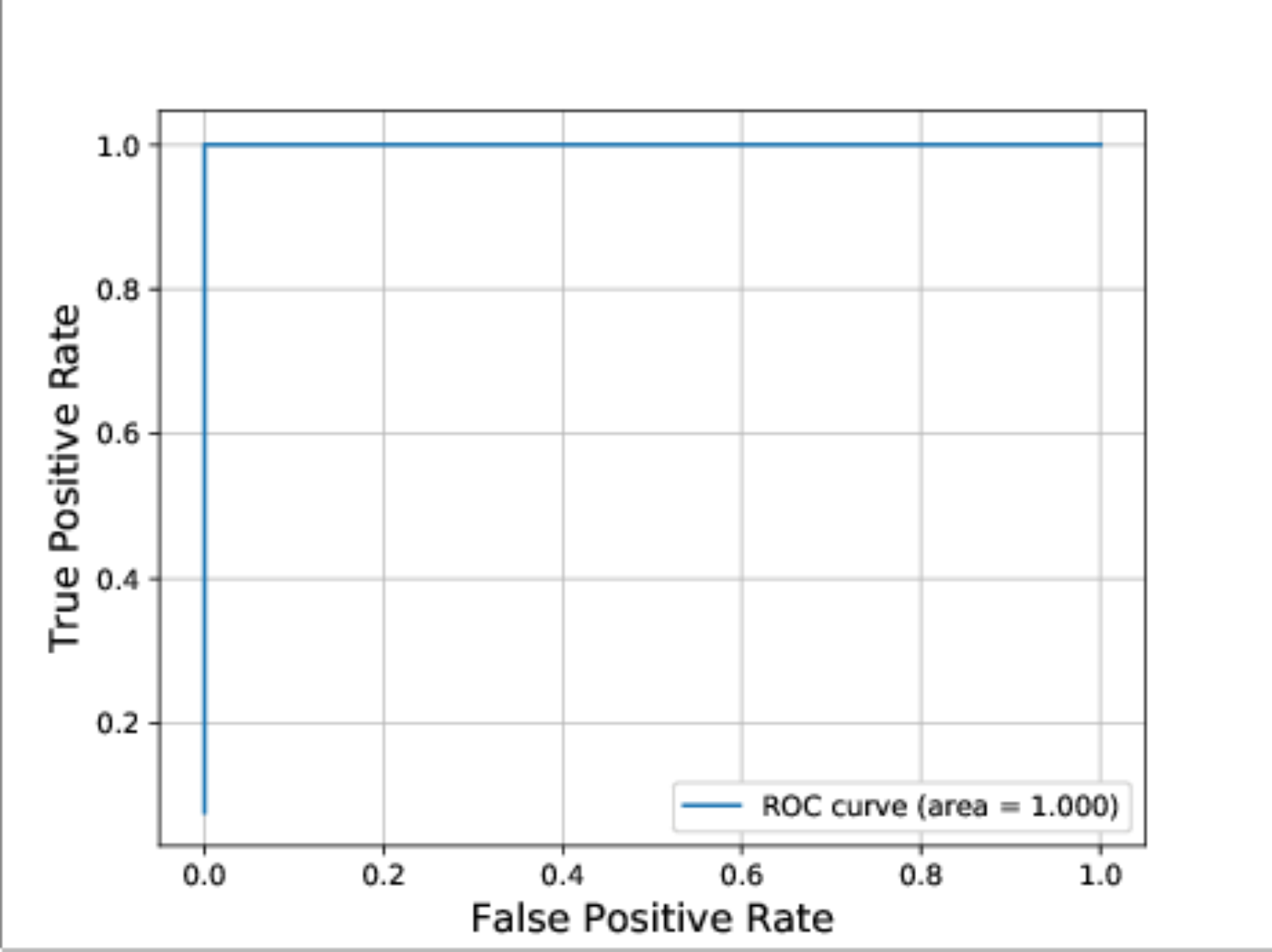} \\
      \hline
\end{tabular}
\caption{UNM dataset: histogram of the covering similarity distributions $\mathscr{S}_c(s,S)$ (left column),  ranked covering similarity $\mathscr{S}_c(s,S)$ (middle column), ROC curves (right column), when 10\%  (top row),  22\% (middle row) and 38\% (bottom row) of the normal data is used for training.}
\label{fig:umn_distrib_roc}
\end{figure*}

\begin{figure}
\includegraphics[scale=.5]{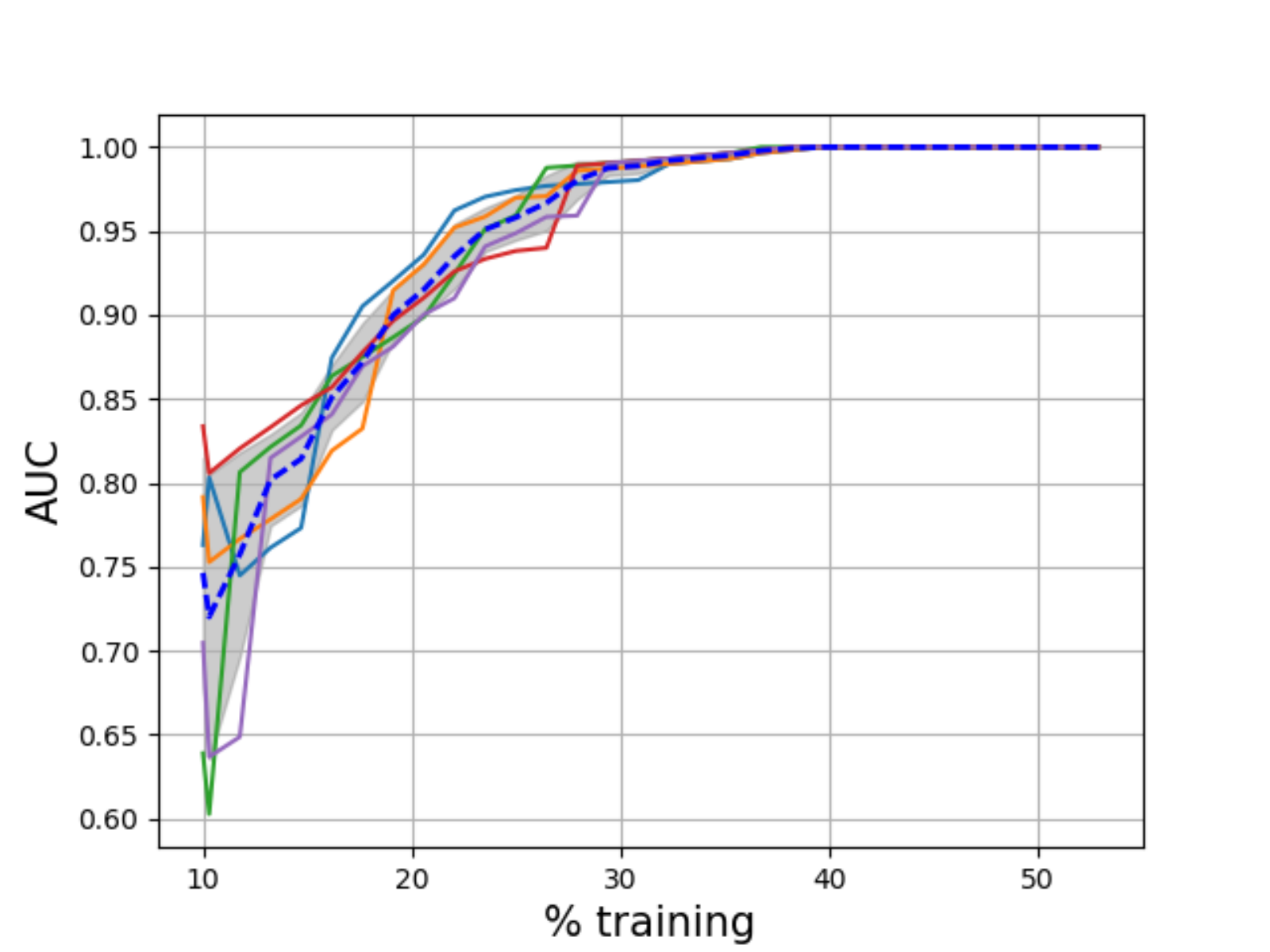}
\caption{UNM dataset: Area Under the ROC Curve (AUC) curves as a function of the size of the training set as a percentage for 5 distinct runs. The average curve is shown as a dotted line. The grey filled area shows the +1/-1 standard deviation curves.}
\label{fig:UNM_Auc-vs-trainSize}
\end{figure}

\begin{figure}
\includegraphics[scale=.5]{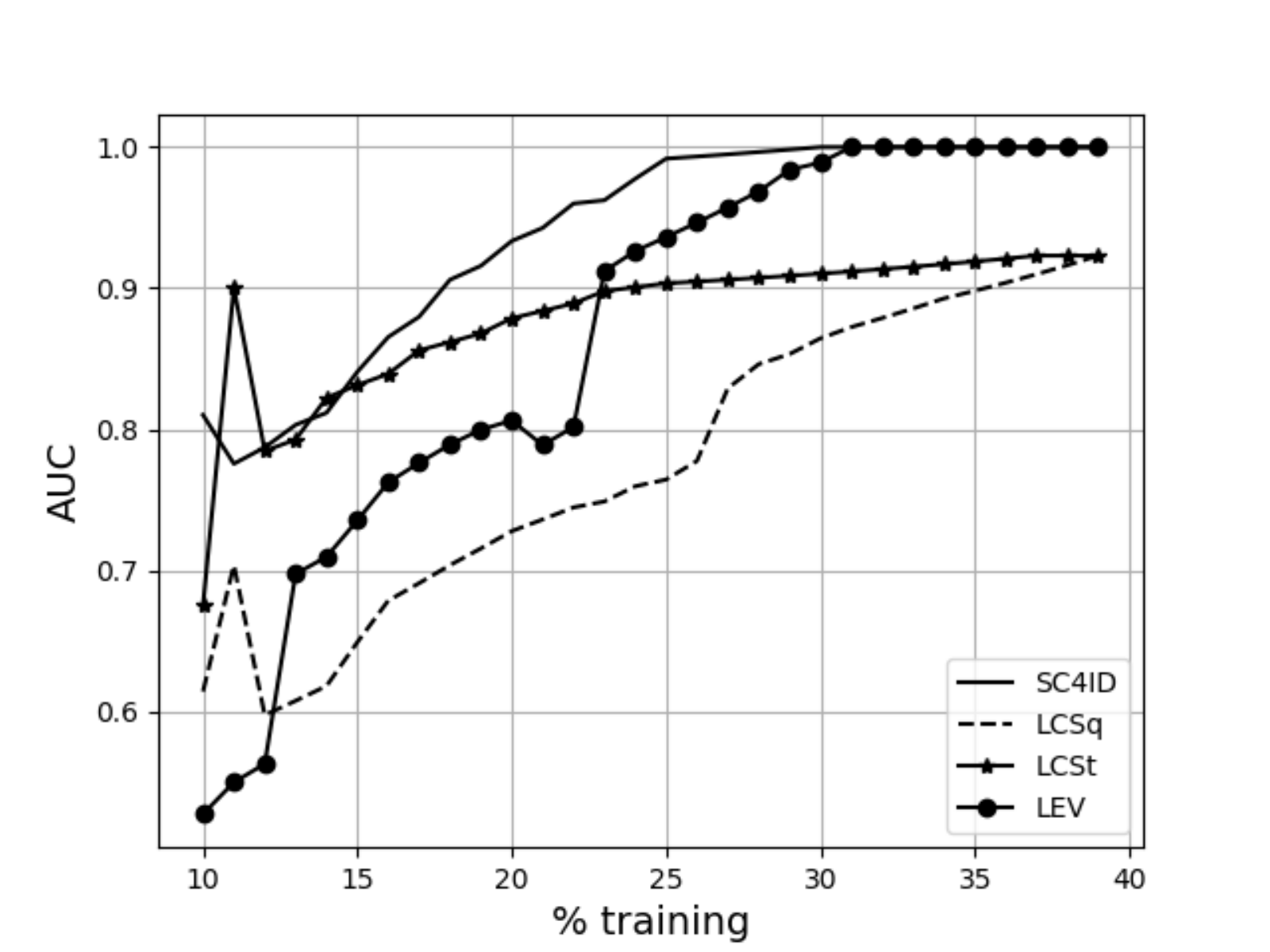}
\caption{UNM dataset: AUC curves as a function of the size of the training set as a percentage for SC4ID (plain curve), LEV (circle), LCSt (star), and LCSq (dotted line) similarity measures.}
\label{fig:UNM_Aucs-vs-trainSize}
\end{figure}

\section{Experiment}

\begin{table}[!h]
\begin{center}
\begin{tabular}{|c|c|c|}
\hline
 & UNM & ADFA-LD\\
 \hline \hline
 $|S|$ & 81 & 5951 \\
 \hline
 min seq. length & 7 & 75 \\
 \hline
 max seq. length & 183018 & 4494 \\
 \hline
 mean seq. length& 9031 & 462 \\
 \hline
 std seq. length & 25111 & 522\\
 \hline
\end{tabular}
\end{center}
\caption{Some statistics on the UNM and ADFA-LD datasets.}
\label{tab:stats-UNM-ADFA-LD}
\end{table}

We have evaluated the SC4ID algorithm on two well-known system call datasets provided respectively by the University of New Mexico (UNM) \cite{UNM1998} and the Australian Centre for Cyber-Security, referred to as ADFA-LD \cite{Creech2013}. Some statistics summarizing the content of these datasets are given in Table \ref{tab:stats-UNM-ADFA-LD}. We have selected these two datasets because i) their use by the research community in intrusion detection covers a wide time span (1998-2013) and ii) numerous results have been reported using these datasets in various settings, which enable to compare our findings to the state of the art. \\

We have used as baselines, three  similarity measures commonly used in text processing or bioinformatics, namely the Levenshtein's distance (LEV) \cite{Levenshtein66, Wagner1974}, the Longest Common Subsequence (LCSq) \cite{Maier1978, Bergroth2000} and the Longest Common Substring (LCSt) \cite{Gusfield1997}.\\

The longest common substring problem, also known as the longest common contiguous subsequence problem, is to find the longest string that is a substring of two or more strings. The pairwise similarity based on the longest common substring that we used is normalized as follows :
\begin{equation}
\mathcal{S}_{LCSt} (s_1,s_2) = \frac{LCSt(s1,s2)}{max(|s1|,|s2|)}
\end{equation}
where $LCSt(s1,s2)$ is the length of the longest common substring of $s_1$ and $s_2$. An implementation using suffix-trees leads to a $O(n+m)$ algorithmic complexity, where $m$ and $n$ are the lengths of the pair of strings or sequences that are compared.
 
The Longest Common Subsequence is similar to a longest common substring problem, except that the symbols of the longest subsequences do not need to be contiguous in the input sequences, namely gaps are authorized. The pairwise similarity based on the longest common subsequence that we used is normalized as follows :
\begin{equation}
\mathcal{S}_{LCSq} (s_1,s_2) = \frac{LCSq(s1,s2)}{max(|s1|,|s2|)}
\end{equation}
where $LCSq(s1,s2)$ is the length of a longest common subsequence of $s_1$ and $s_2$.
Using a dynamic programming implementation, the algorithmic complexity of LCSq distance is $O(nm)$, where $m$ and $n$ are the lengths of the pair of strings or sequences that are compared. The LCSq distance is conceptually close to the Smith and Waterman \cite{Smith1981} distance used in Bioinformatics.\\
 
The Levenshtein's distance, also known as the edit distance, is the minimum number of single-symbolic edit operations (insertions, deletions or substitutions) required to change one string or sequence into the other. Using a dynamic programming implementation, the algorithmic complexity of the Levenshtein's distance is $O(nm)$, where $m$ and $n$ are the lengths of the pairs of strings or sequences that are compared. The Levenshtein's distance is conceptually close to the Needleman Wunsch \cite{Needleman1970} distance used in Bioinformatics. \\
The pairwise similarity based on the Levenshtein's distance that we used is normalized as follows :
\begin{equation}
\mathcal{S}_{LEV} (s_1,s_2) = 1 - \frac{LEV(s1,s2)}{max(|s1|+|s2|)}
\end{equation}
where $LEV(s1,s2)$ is the Levenshtein's distance between $s_1$ and $s_2$.\\

The averaged costs for the computation of the similarity matrix required to solve the HIDS problem are
\begin{itemize}
\item $O(k\cdot n \cdot log n)$ for $\mathcal{S}(s1,S)$
\item $O(N\cdot n)$ for $\mathcal{S}_{LCSt}$
\item  $O(N\cdot n^2)$ for $\mathcal{S}_{LCSq}$ and $\mathcal{S}_{LEV}$
\end{itemize}
where $N$ is the size of the normal sequences used as training, $n$ is the average length of the training and test sequences and $k$ is the average length of the $S$-optimal coverings of the test sequences.

To assess the experimented algorithms on the binary classification task that consists of separating attacks from normal data, we use the \textit{receiver operating characteristic} (ROC) curve that shows the detection rate (true positive rate) as a function of the false alarm rate (false positive rate) as the discrimination threshold of the classifier is varied. We use the \textit{Area Under the ROC curve} (AUC) as the global assessment measure.

\subsection{The UNM dataset}

For our first experiment, we have considered the sendmail system call traces from the relatively old UNM dataset \cite{Hofmeyr1998}. These synthetic data were collected at UNM on Sun SPARCstations running unpatched SunOS 4.1.1 and 4.1.4 with the included sendmail.  We have adopted the setting described in  \cite{Yolacan2014}, basically $68$ unique process traces in the normal dataset (uniqueness is to ensure that no sequence appears simultaneously in the training and validation subsets) and $13$ process traces in the abnormal dataset (we have kept the $3$ duplicated attack sequences in this set).

The experimental protocol for this data set and the covering similarity is as follows
\begin{enumerate}
\item \underline{Initialization}: randomly select $10\%$ of the normal data to build the 'normal' model, i.e. the initial set of normal data, $S$, from which the tested similarity will be evaluated.
\item \underline{Evaluation}: for each of the remaining normal and attack sequences, $s$ , evaluate the covering similarity $\mathscr{S}(s,S)$. Then rank the normal data according to their similarity score. Finally evaluate the ROC curve, and assessment measure (AUC).
\item \underline{Selection}: from the previously ranked normal data, select the normal sequence with the worst similarity score, $s^-$, and update the normal model with training set $S := S \cup \{s^-\}$. Then loop in step $2$ until $50\%$ of the normal data is used in training. \\
\end{enumerate}

The same protocol is used for the other tested similarities, except that the normalized distance to the closest 'normal' sequence is used instead of $\mathscr{S}(s,S)$.\\

Fig. \ref{fig:umn_distrib_roc} presents the histogram of the covering similarity values (left column) for the 'normal' (blue) and 'attack' (red) data, the ordered  similarity values in increasing order (middle column) for the 'normal' (blue dotted line) and attack data (red, continuous line) and the ROC curve (right column). In this figure, the top row corresponds to a situation for which $10\%$ of the normal data has been randomly selected for training, the middle row where the initial $10\%$ of training data has been enriched using $12\%$ of the remaining normal data corresponding to the lowest covering similarity scores, and finally, the bottom row corresponds to an enrichment of the initial $10\%$ of training data using $28\%$ of the remaining normal data. From top to bottom, we show that the model improves its capacity to separate attack data from normal data: the AUC value that is initially $.81$, reaches $.94$ when $22\%$ of the 'normal' data is used as training and $1.0$ when 38\% of the 'normal' data is used as training. In the middle column, we see that the similarity scores for the normal data is progressively tangent to the $1$ constant curve (maximal similarity value), while, for the attack data, it generally stays much lower, although it increases slightly when the size of the training set increases.\\

Figure \ref{fig:UNM_Auc-vs-trainSize} presents the AUC value obtained by the SC4ID algorithm as the enrichment of the training data increases. $5$ different runs have been carried out, each one being initialized by randomly selecting $10\%$ of the normal data. The grey area corresponds to the +1/-1 standard deviation, and the blue dotted line is the average curve. We can see on this figure that the SC4ID algorithm improves quite rapidly until reaching a perfect separation of the normal and attack data when $38\%$ of the normal data is used, regardless of the initialization. This is a major and very promising result, since the instance selection that is performed directly from the covering similarity scores is working particularly well on the UNM data.  \\

Figure \ref{fig:UNM_Aucs-vs-trainSize} compares the 4 tested algorithms on a single run . In the Figure, one can see that the SC4ID algorithm (plain curve) is the first to reach a perfect AUC value when about 30\% of the available normal data is used as training. The Levenshtein similarity based algorithm (LEV) is the second to reach a perfect AUC value when about 31\% of the available normal data is used as training data. The LCSq algorithm reaches an asymptotic AUC value of $.923$ and is unable to separate normal data from attack data. Due to the presence of gaps in  the longest common subsequences that are extracted, some attacks become too similar to normal data and perfect separation is no longer possible. We deduce from this result that the Smith and Waterman distance used in bioinformatics will not perform well either in this classification task.  Finally, the LCSt distance will reach a perfect AUC value when 57\% of the available normal data is used as training data, which is almost twice the size of training data than it was necessary for SC4ID or LEV.\\

\begin{table}
\begin{center}
\begin{tabular}{c|c}
 & \textbf{Elapsed time} \\
 \textbf{Similarity} & \textbf{(sec.)}\\
\hline\hline
SC4ID & 454\\
LCSt & 653 \\
LCSq(*) & 27023\\
LEV(*) & 37296
\end{tabular}
\end{center}
\caption{Elapsed time for the four tested algorithms. SC4ID: Covering Similarity, LCSt: Longest Common Substring, LCSq: Longest (non contiguous) Common Subsequence, LEV: Levenshtein. (*) note that for LCSq and LEV, 20 processor cores have been used.}
\label{tab:UNM-elapsedTime}
\end{table}

Table \ref{tab:UNM-elapsedTime} gives the average elapsed time for each enrichment iteration. We can clearly see that SC4ID is the faster algorithm, LCSt takes 44\% more time and LCSq and LEV take more than 6000\% more time than SC4ID. Note that LEV and LCSq were running on a 20 cores architecture while SC4ID and LCSt  have not been parallelized, due to their suffix-tree implementation, and were running on a single core.

\subsection{The ADFA-LD dataset}
The second  experiment involves a much more recent benchmark dataset for HIDS assessment. Yhis benchmark has been designed by the Australian Centre OF Cyber-Security (ACCS). 
According to the authors, the ADFA-LD dataset \cite{Creech2013, Creech2014} has been designed to reflect modern hacking techniques, whilst using modern patched software as its backbone. It is composed with a training set composed with $833$ normal sequences, a validation set composed with $4373$ normal sequences and a set of $746$ attack sequences partitioned into $6$ categories referred to as HydraFTP, Hydra-SSH, Adduser, Java-Meterpreter, Meter-preter and Webshell.  
 
Similarly to the UNM experiment, the experimental protocol for this data set is as follows :
\begin{enumerate}
\item Initialization: select the 833 sequences of the normal training data to build the 'normal' model, i.e. the initial set of normal data, $S$, against which the covering similarity will be evaluated.
\item Evaluation: for all the remaining normal and all attack sequences, $s$, evaluate the covering similarity $\mathscr{S}(s,S)$. Then rank the normal data according to their similarity scores. Finally evaluate the ROC curve and assessment metrics.
\item From the previously ranked normal data, select 100 sequences of normal data with the worst similarity score, $S_{100}$, update the normal model $S := S \cup S_{100}$ and loop in step $2$ until $50\%$ of the normal data is used in training. 
\end{enumerate}

\begin{figure}
\includegraphics[scale=.5]{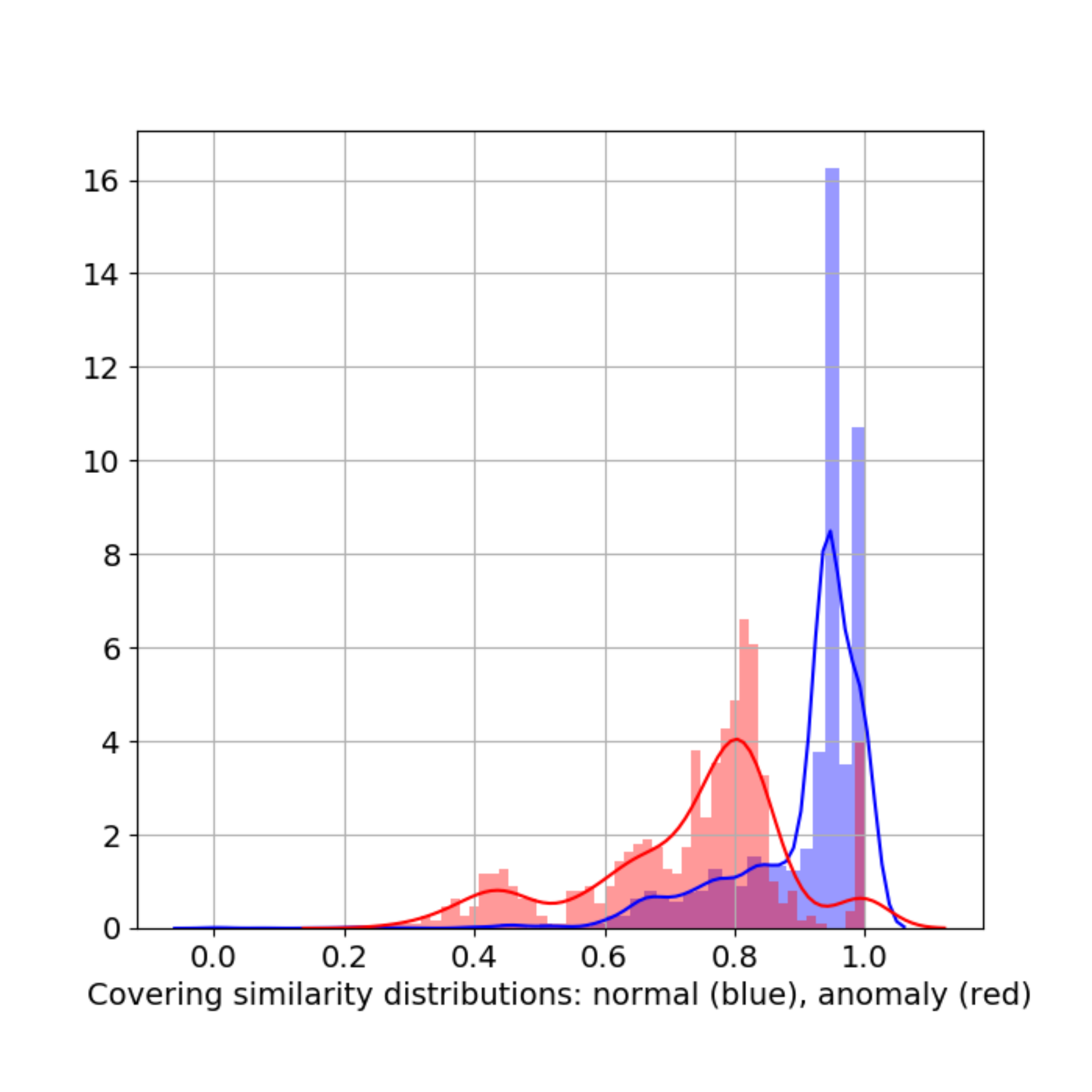}\\
\includegraphics[scale=.5]{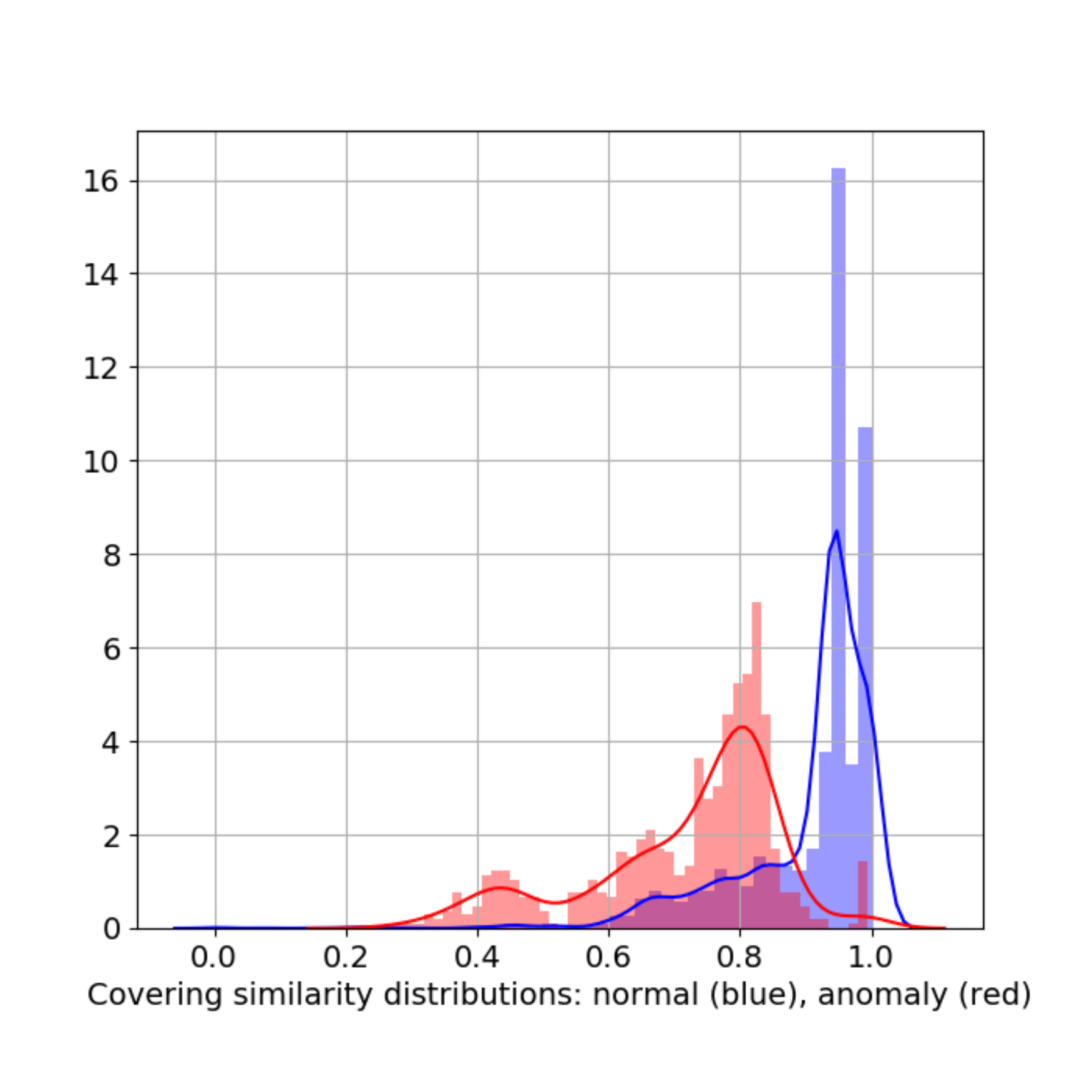}
\caption{Histogram of the covering similarity for normal validation data (blue), and attack data (red). Top: all the attacks are considered. Notice the red pick with maximal similarity, $1.0$, corresponding to attacks that are exact subsequences of the train sequences. Bottom: the same histogram when the $32$ attacks with covering similarity equal to $1.0$ are removed. Notice the remaining small peak of high similarity close to $1.0$}
\label{fig:ADFA-LD_HIST833}
\end{figure}
 
\begin{figure*}
\begin{tabular}{|ccc|}
      \hline
      \includegraphics[width=60mm, height=60mm]{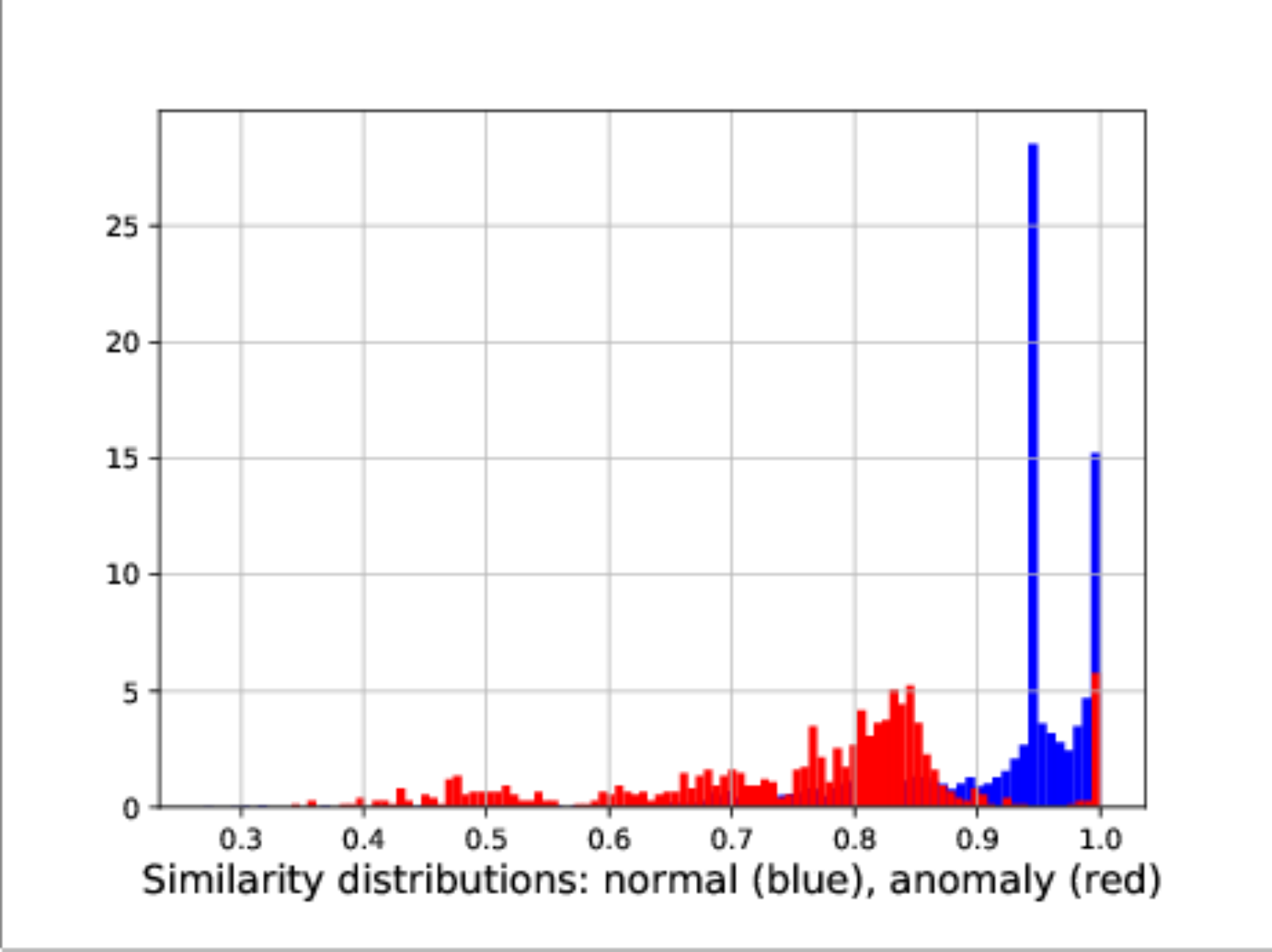} &
      \includegraphics[width=60mm, height=60mm]{coveringSimilarity_0.pdf} &
      \includegraphics[width=60mm, height=60mm]{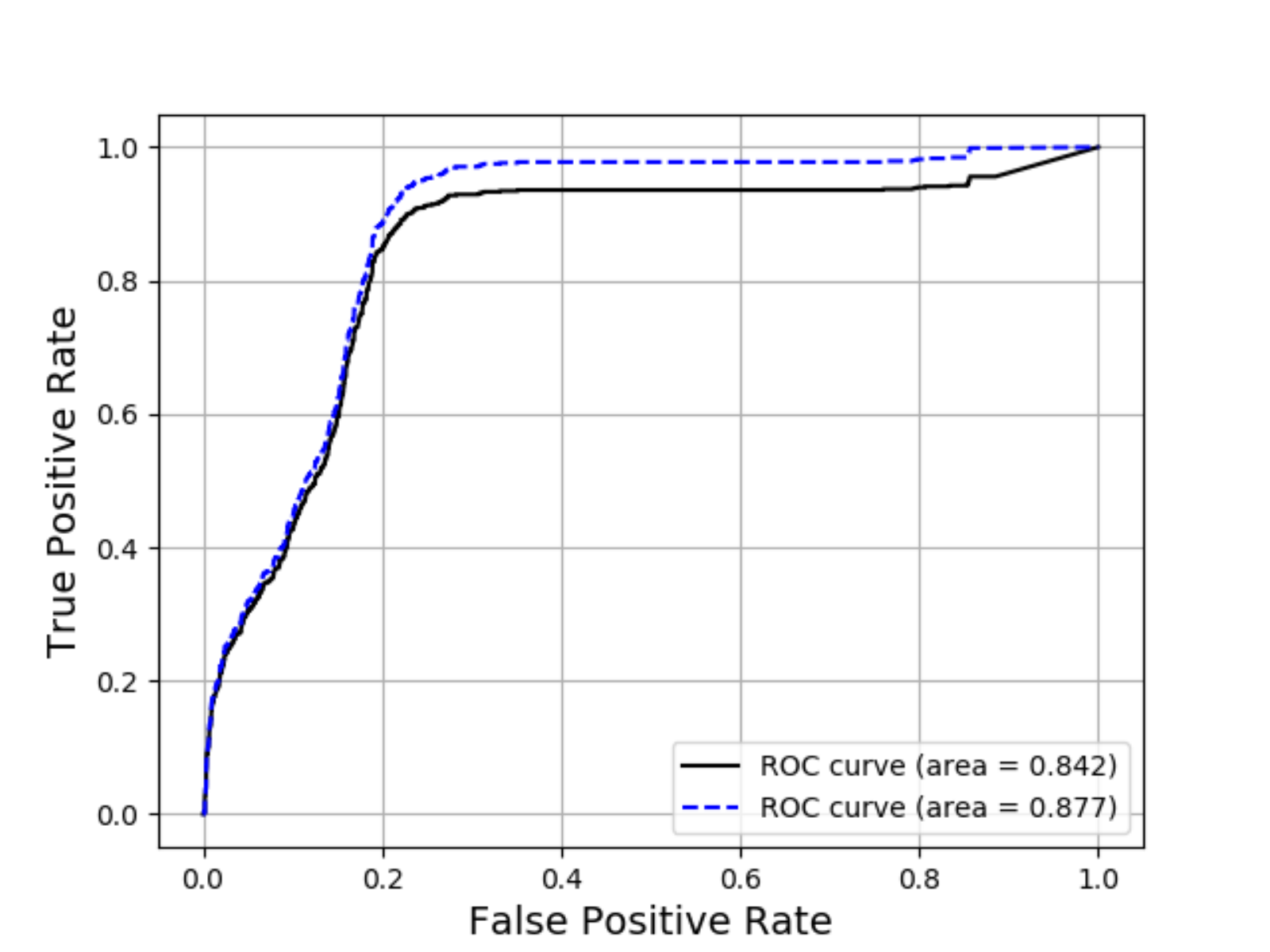} \\
      \hline
      \includegraphics[width=60mm, height=60mm]{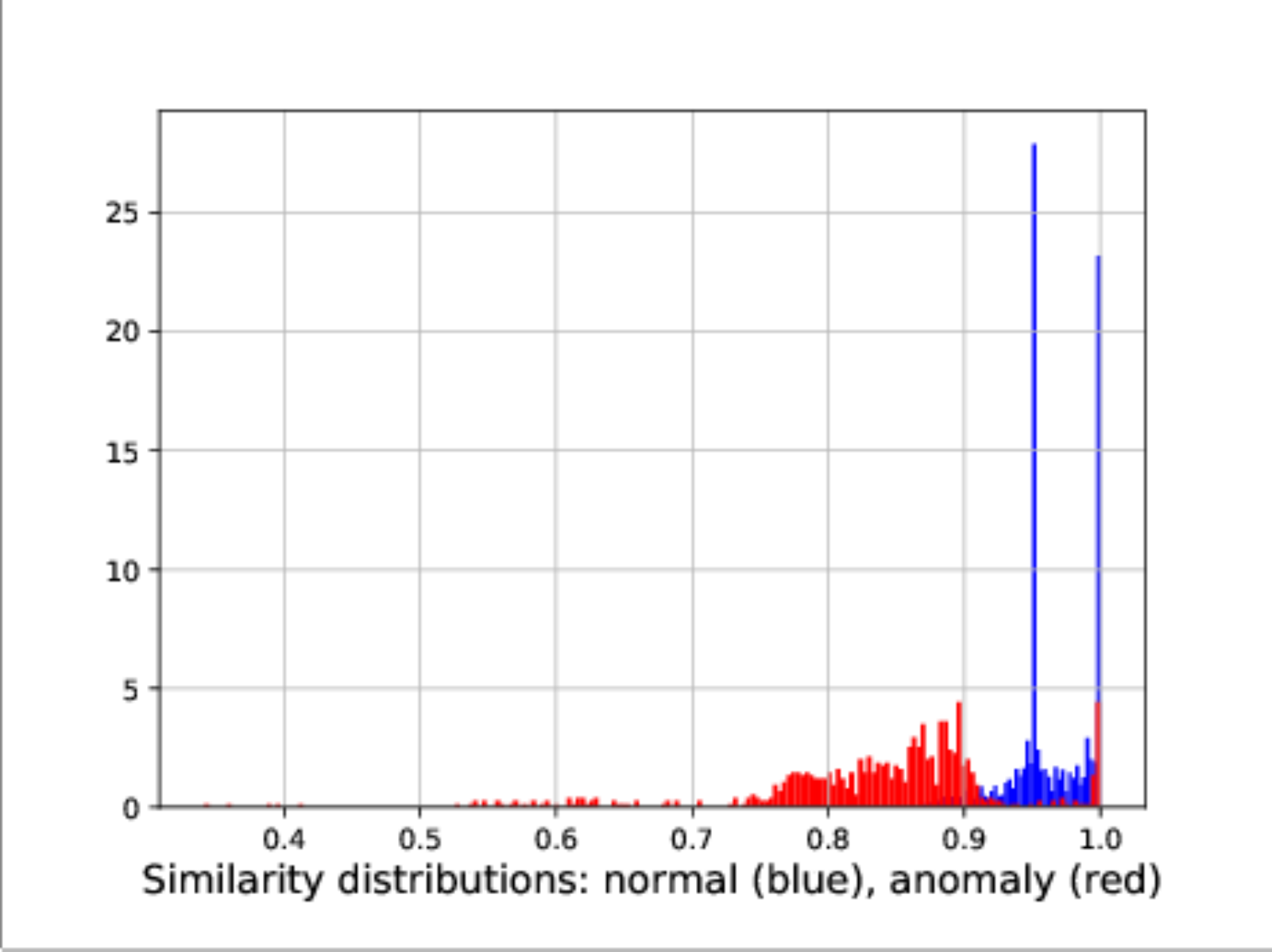} &
      \includegraphics[width=60mm, height=60mm]{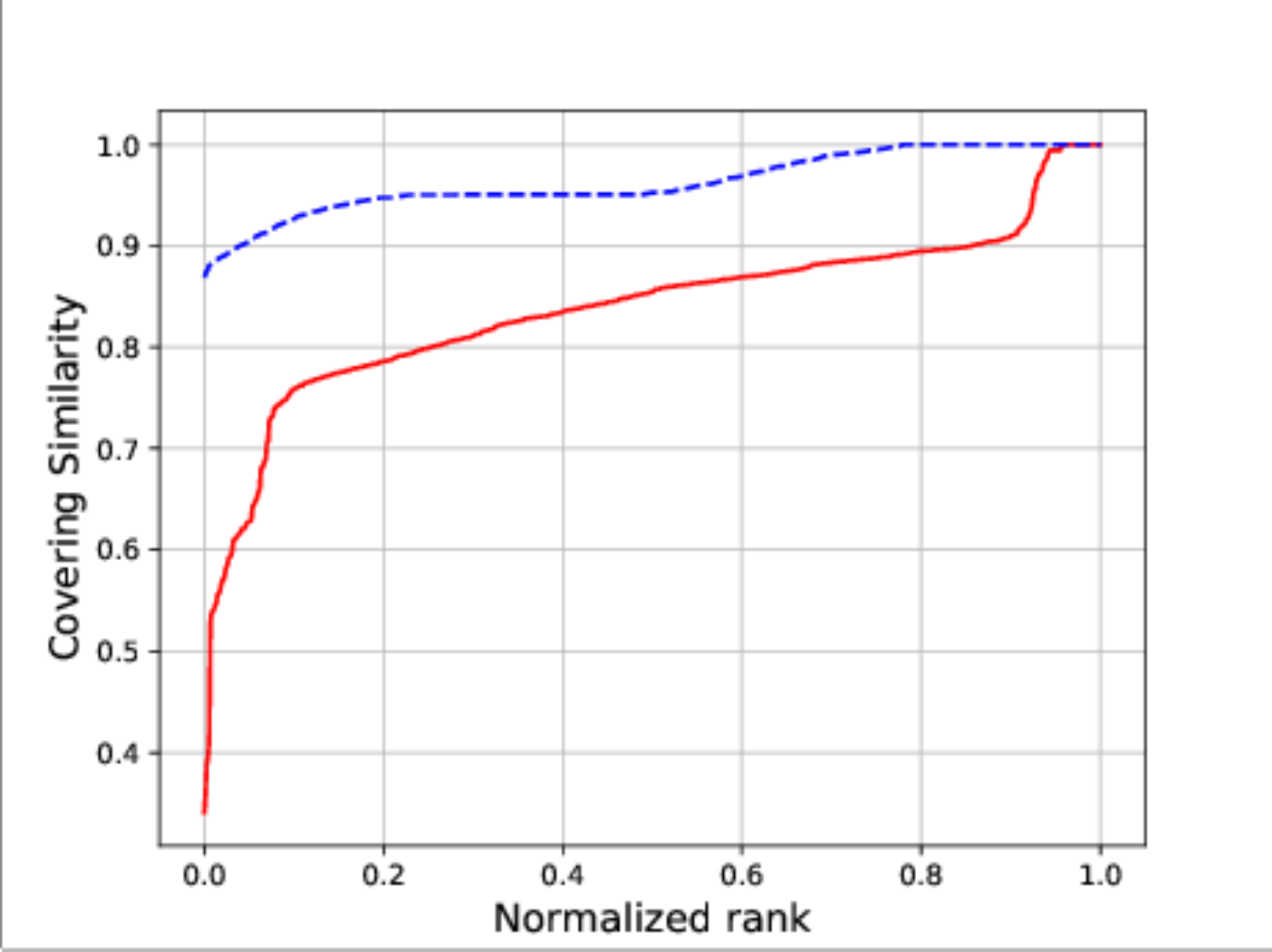} &
      \includegraphics[width=60mm, height=60mm]{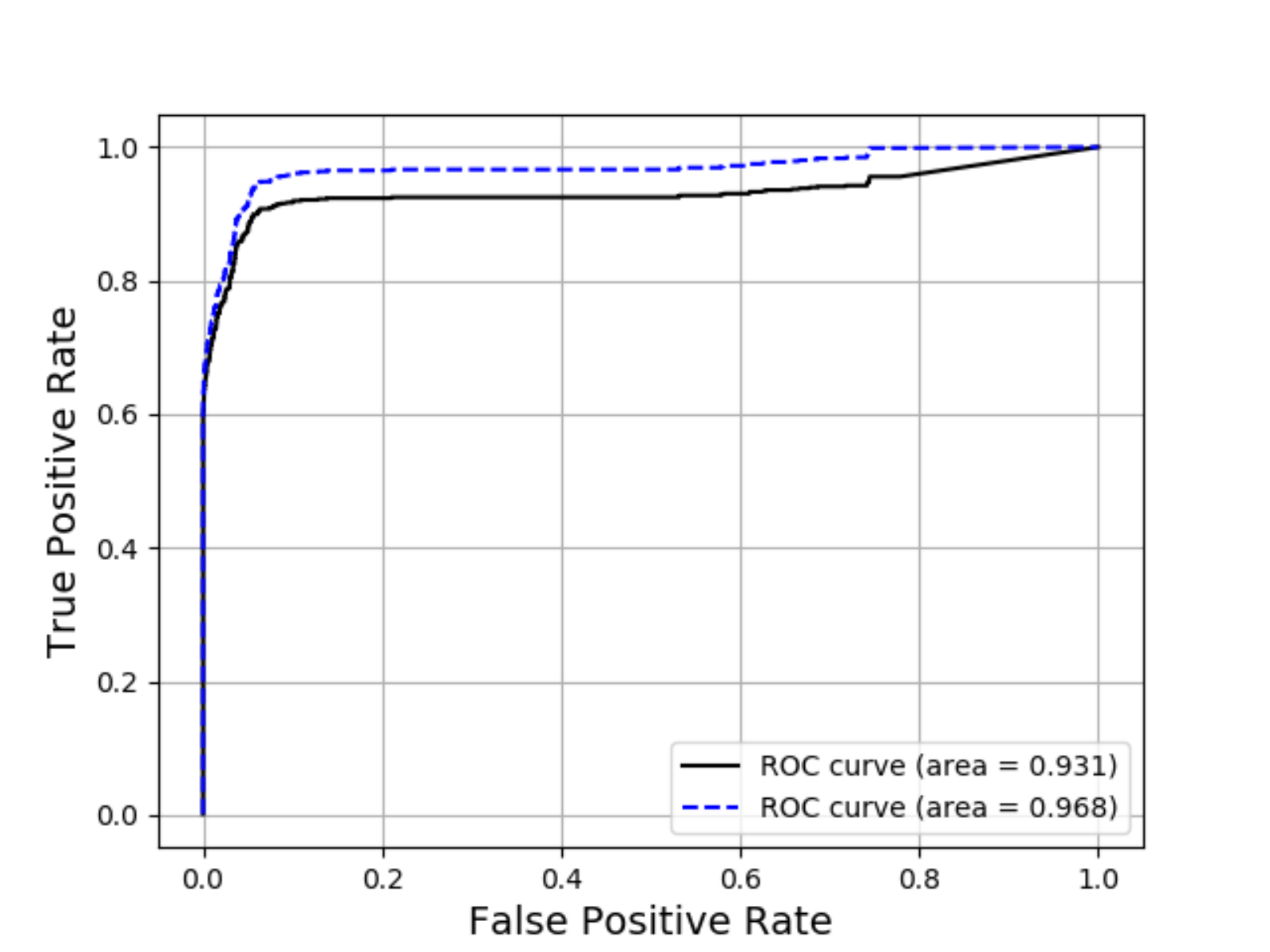} \\
      \hline
      \includegraphics[width=60mm, height=60mm]{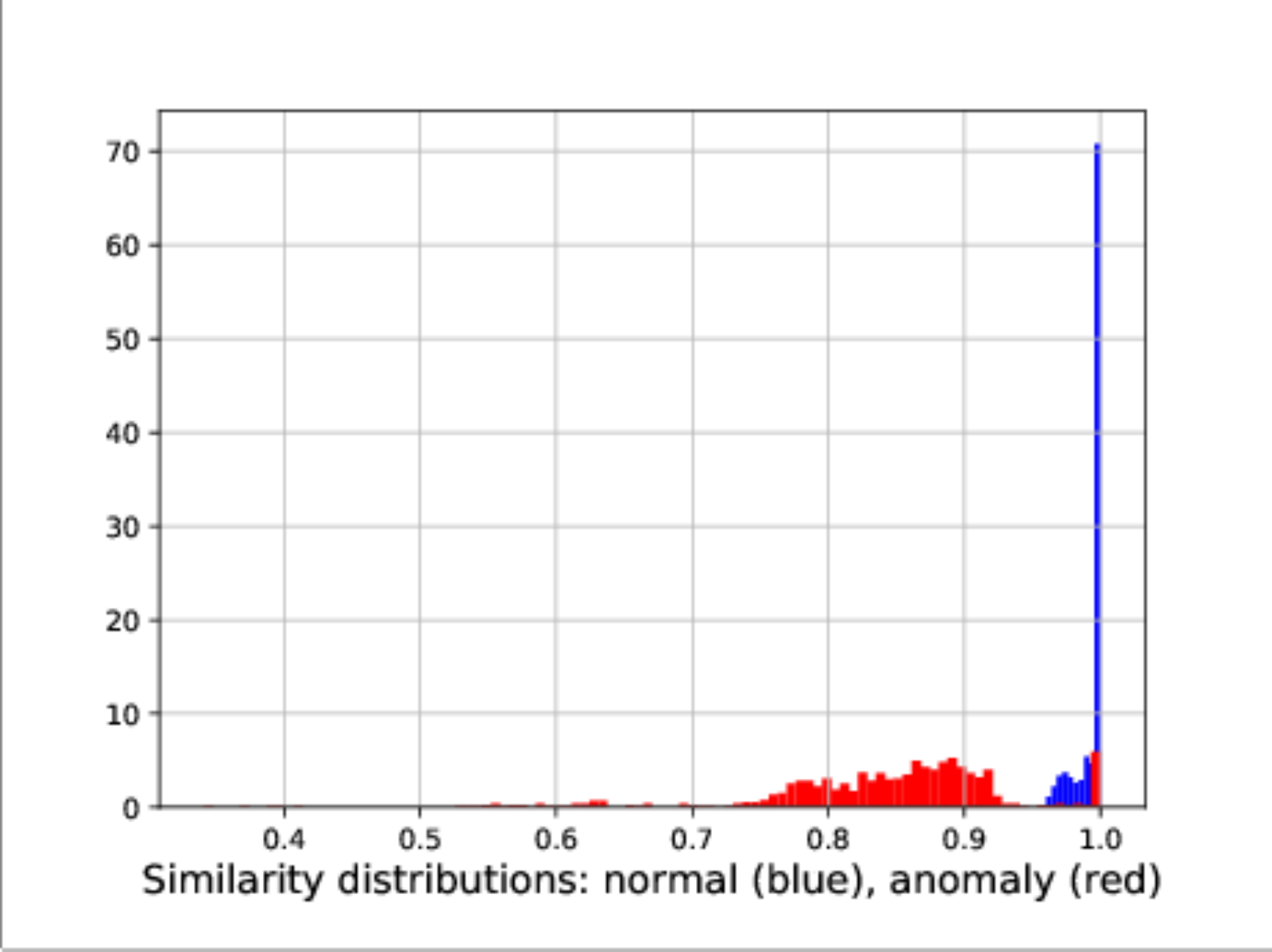} &
      \includegraphics[width=60mm, height=60mm]{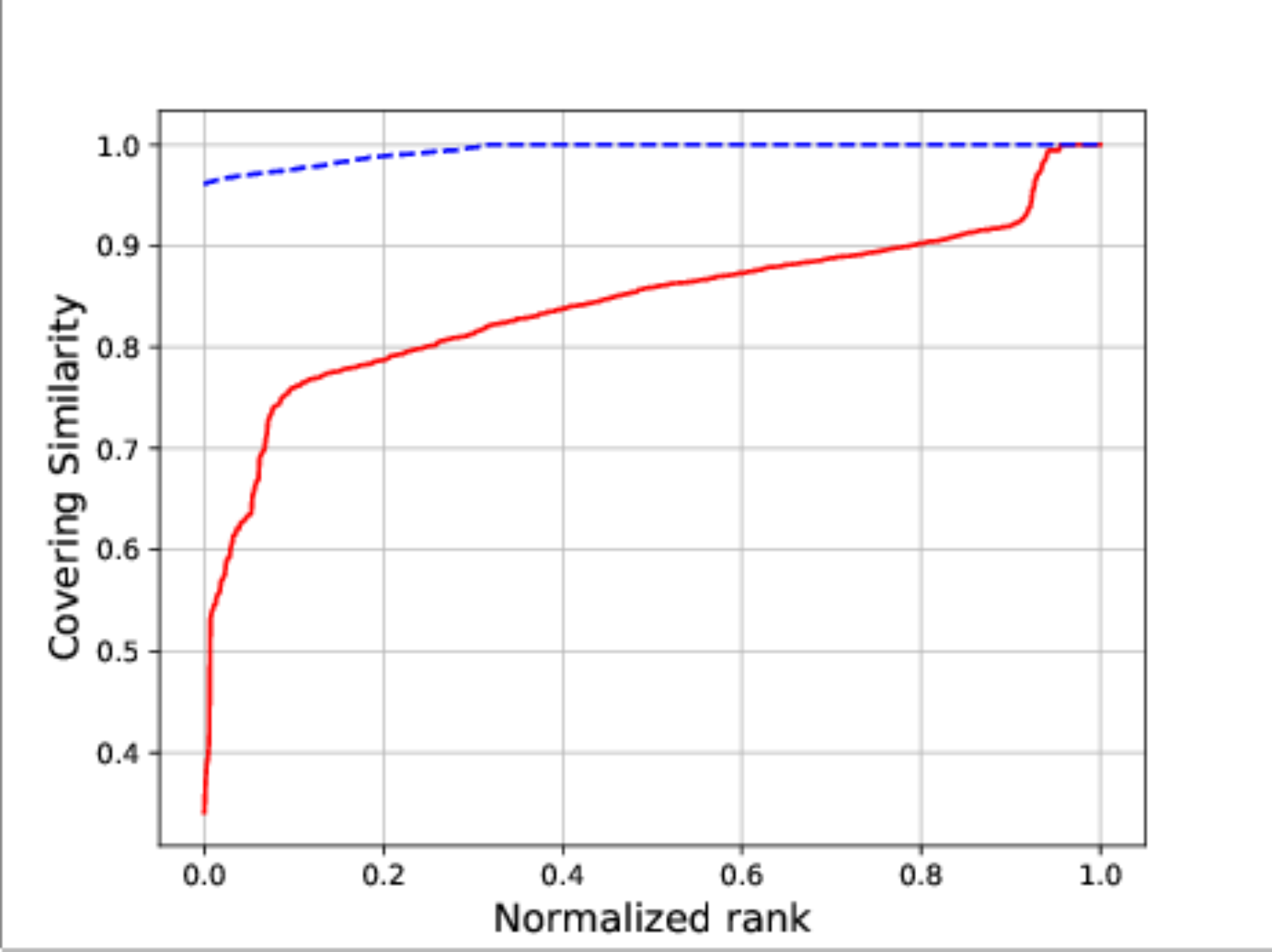} &
      \includegraphics[width=60mm, height=60mm]{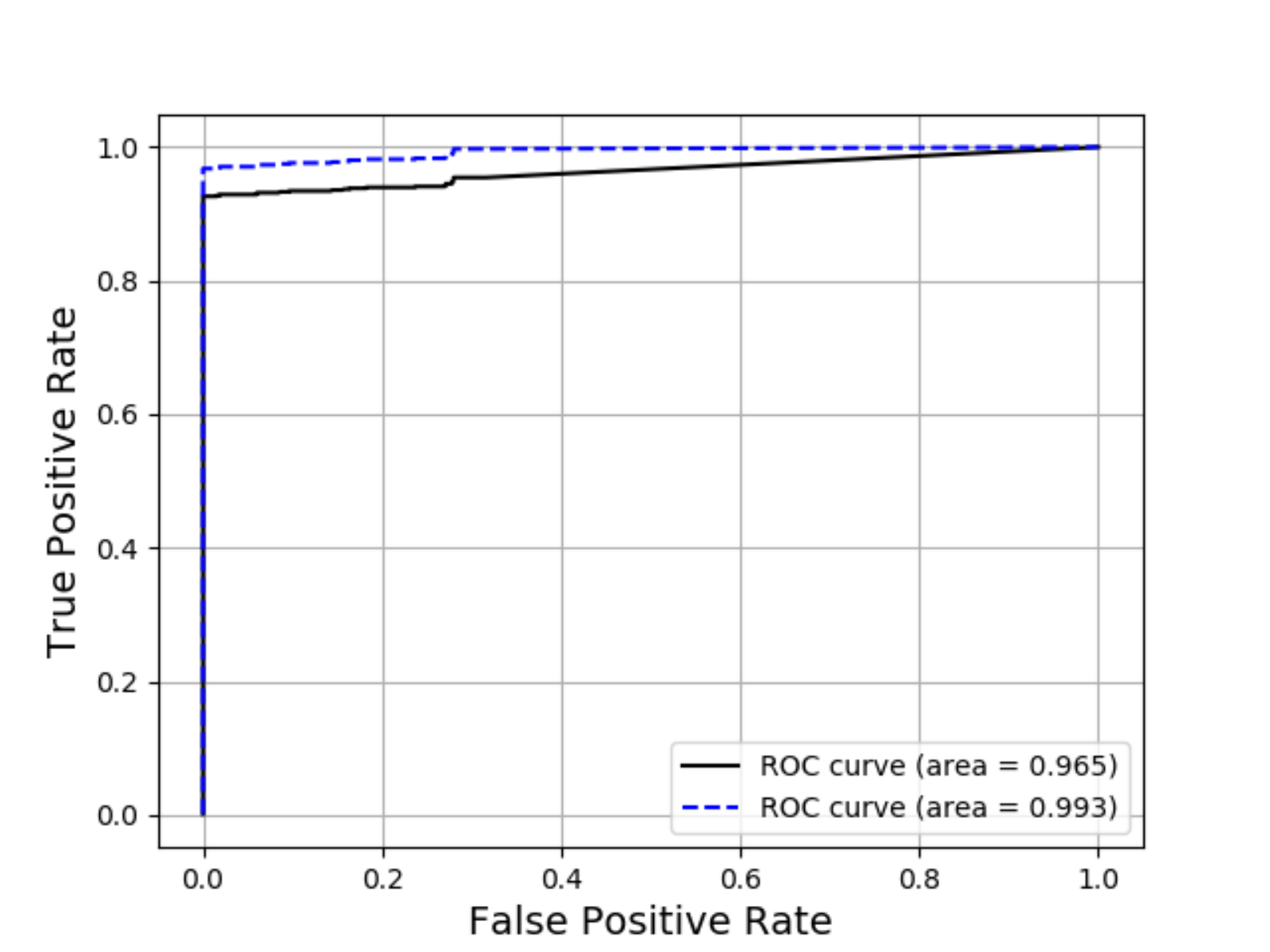} \\
      \hline
\end{tabular}
\caption{ADFA-LD dataset: histograms of the covering similarity distributions $\mathscr{S}_c(s,S)$ (left column),  ranked normalized covering similarities $\mathscr{S}_c(s,S)$ (middle column), ROC curves (right column), when $833$ normal data  (top row), $833$ + $500$ normal data (middle row) and $833$ + $1000$ normal data (bottom row) is used for training.}
\label{fig:adfa-ld_distrib_roc}
\end{figure*}

\begin{figure}
\includegraphics[scale=.5]{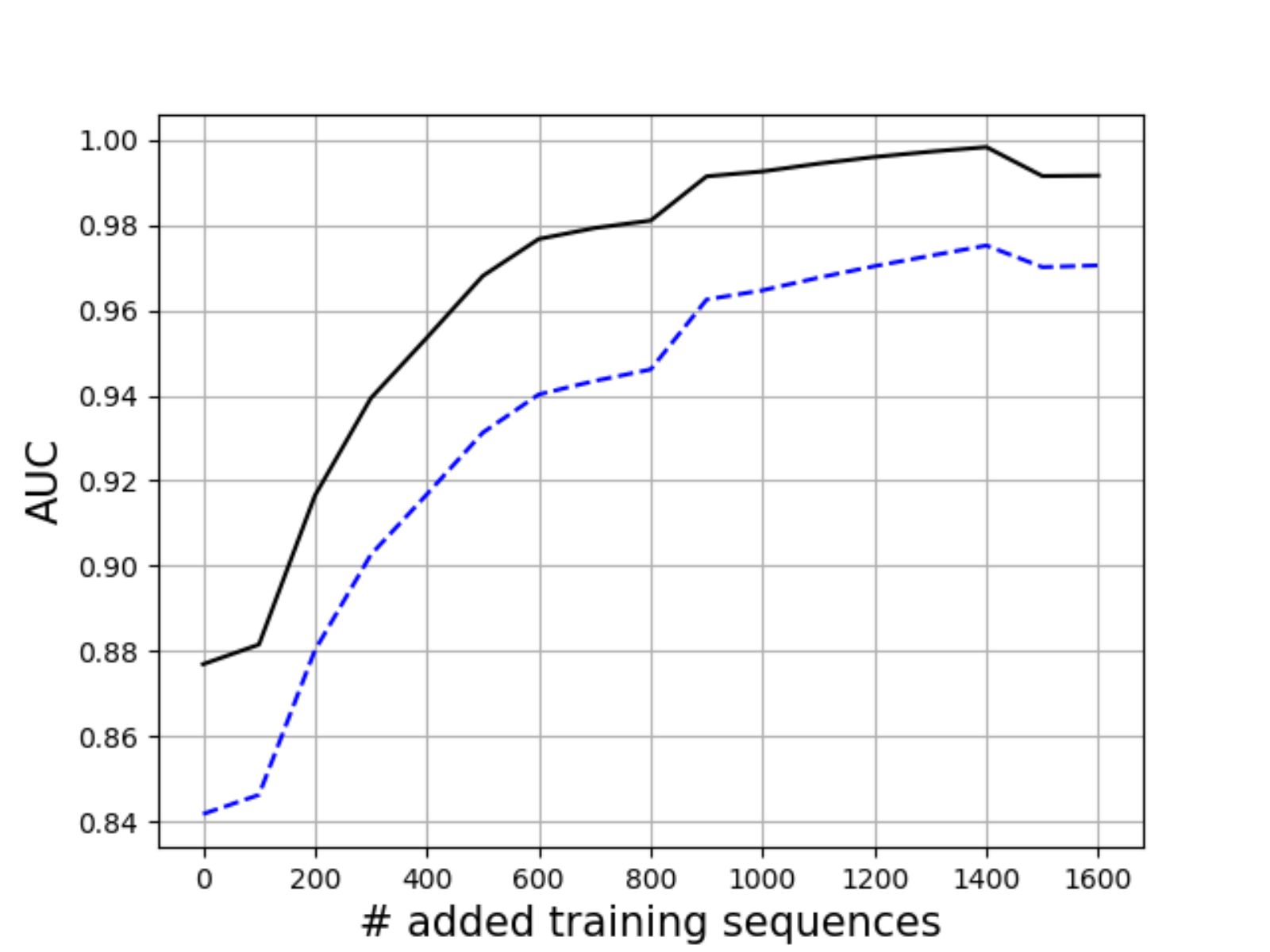}
\caption{ADFA-LD dataset: AUC curves as a function of the number of normal data added to the initial training set. The black continuous line corresponds to the situation for which the sequences of attacks that are an exact subsequence of a normal data have been removed. The blue dotted line corresponds to the case where no sequence of attacks has been removed.}
\label{fig:adfa-ls_Auc-vs-trainSize}
\end{figure}

When only $833$ normal training sequences are used to train SC4ID, one acquires the histogram of similarity scores given in the top of Fig. \ref{fig:ADFA-LD_HIST833}. A red peak of maximal similarity ($1.0$) exists, meaning that some attack data (actually $32$ sequences) are exact subsequences of the 833 normal training data. When we remove these $32$ attack sequences, one observes a small remaining red peak with high covering similarity close to $1.0$. This peak corresponds to 11 attack sequences that have a covering of size $2$, i.e. that are composed of two subsequences that belong to the set of training data.

This situation for which attack sequences can be exact subsequences of normal data is a bit surprising. 
We do not have to investigate the nature of such subsequences. 
Rather, we take the opportunity to highlight an obvious limitation of SC4ID: it is completely blind to this kind of situation.  Complementary approaches are required to specifically address this issue if any such situation exists.

Fig. \ref{fig:adfa-ld_distrib_roc} presents the histogram of the covering similarity values (left column) for the 'normal' (blue) and 'attack' (red) data, the ordered  similarity values in increasing order (middle column) for the 'normal' (blue dotted line) and attack data (red, continuous line) and two ROC curves (right column): the black continuous line curve corresponds to the situation where all the attack data is retained, the blue dotted line corresponds to the situation where the $32$ attack sequences that are exact subsequences of the $833$ normal training data have been removed. In this figure, the top row corresponds to the situation where only the $833$ normal training sequences are used for training, the middle row where the initial $833$ sequences have been enriched using $500$ sequences of the remaining normal data that have the lowest  covering similarity,  and finally the bottom row corresponds to a an enrichment of the initial 833 training sequences using the $1000$ sequences of the remaining normal data having the lowest covering similarity. From top to bottom, we show that the model improves its capacity to separate attack data from normal data: the AUC value that is initially $.84/.88$, reaches $.93/.97$ when $500$ normal sequences have been added to the initial training data and $.96/99$ when $1000$ normal sequences have been added to the $833$ initial training set. Similarly as observed for the UMN dataset, in the middle column, we see that the similarity score for the normal data is progressively tangent to the $1$ constant curve, while, for the attack data, it generally stays much lower, although it does increase slightly.

Fig. \ref{fig:adfa-ls_Auc-vs-trainSize} presents for this experiment the AUC values for the algorithm as the enrichment of the training data increases. We can see on this figure that the SC4ID algorithm improves rapidly until reaching an almost perfect separation (when the attacks that are subsequences of the training data have been removed) of the normal and attack data when $1400$ normal data have been used to enrich the initial $833$ training set. Here again,  the instance selection that is performed directly from the covering similarity scores is working particularly well. Nevertheless, we notice  after adding $1400$ sequences that the algorithm is a little less efficient with an AUC value that drops from $.975/.998$ to $.971/.992$. The explanation is that a remaining few attacks (precisely the $11$ sequences mentioned earlier that have a covering of size $2$ when using only the $833$ training sequences) become subsequences of some of the last added normal subsequences to the training set. The similarity for these few attacks reaches the maximal value, namely $1.0$, and  they cannot be separated anymore, hence a lower AUC value. 

\subsection{Setting up the $\sigma$ parameter}
\label{subsec:sigma}
The covering similarity being normalized in $[0,1]$, the choice for $\sigma$ is constrained in the unit interval. Fig. \ref{fig:umn_distrib_roc} and Fig. \ref{fig:adfa-ld_distrib_roc} show that, when enough normal training data has been selected using the procedure that consists of selecting first the normal sequences having the lowest similarity score, we can reach an almost perfect separation between attack and normal sequences using $\sigma \lessapprox 1.0$. According to our experiments on system calls, selecting $\sigma=.97$ will lead to the result that small size coverings  (no more than 3\% of the size of the test sequence) will lead to a high detection rate with a very low false alarm rate. It also leads to a relatively small training set since only about 30\% of the available normal data are used for training in this case. We doubt that a general rule for selecting $\sigma$ exists. The fine-tuning of $\sigma$ is application-dependent.

\subsection{Runtime consideration}
\begin{figure}
\includegraphics[scale=.5]{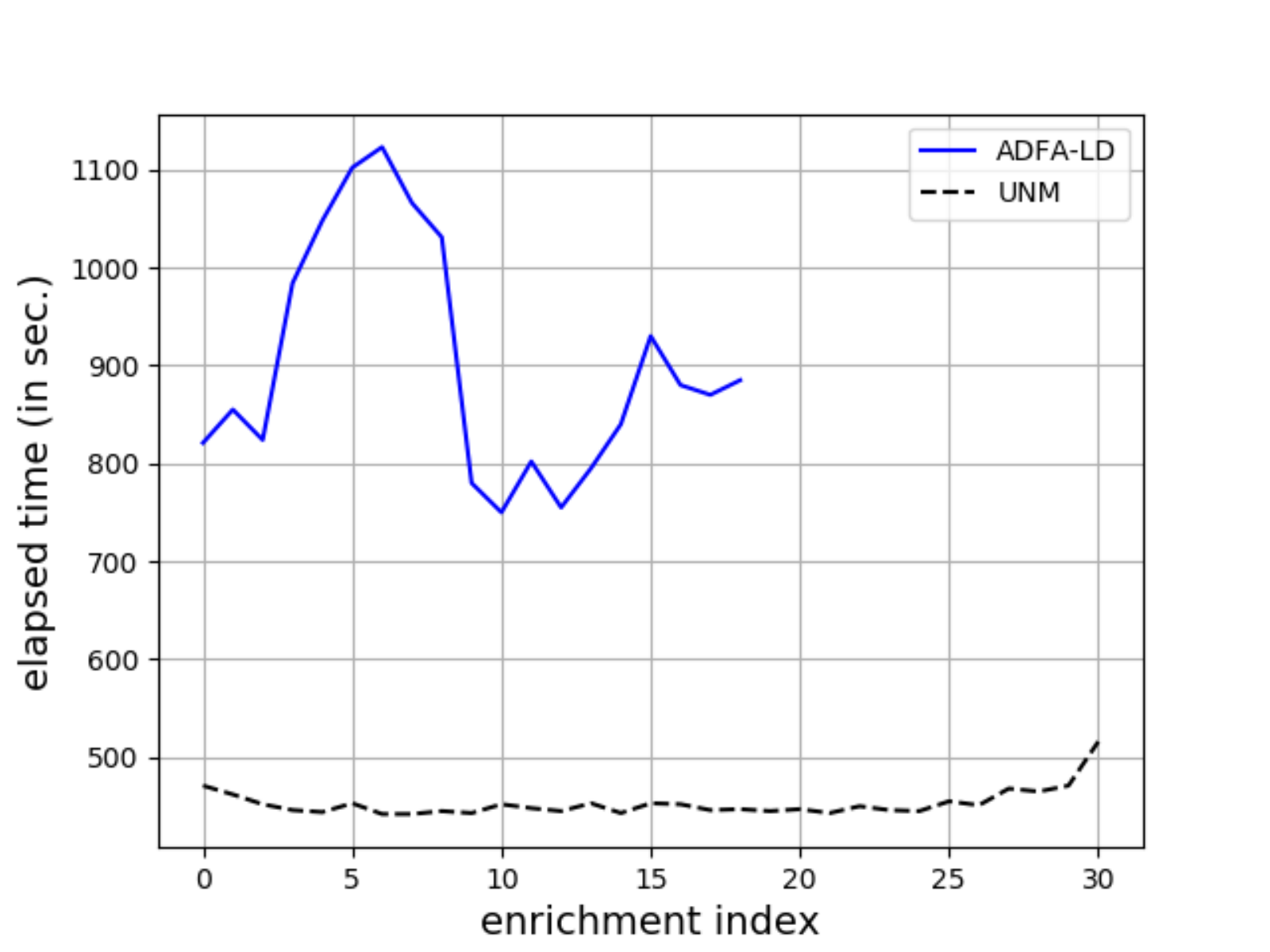}
\caption{Elapsed time in seconds during each step of the enrichment procedure for the UNM (dotted black line) and ADFA-LD (continuous blue line) datasets.}
\label{fig:elapsed-time}
\end{figure}

We have not been able to provide a comparative study on the ADFA-LD dataset, due to the too high time complexity of the three other similarities (LEV, LCSt, LCSq). In particular, given the dataset statistics presented in Table \ref{tab:stats-UNM-ADFA-LD} and the quadratic complexity of LEV and LCSq measures, we estimate that,  if we use these two measures, it would require about 5 months of computation for achieving the results on a simple core of Intel i7 architecture, while it takes about 15 minutes for SC4ID using the same hardware. Consequently, we only discuss hereafter runtime matters for the SC4ID algorithm.

Fig. \ref{fig:elapsed-time} presents the elapsed time expressed in seconds for each iteration of the enrichment process carried out by SC4ID for the UNM (dotted curve) and ADFA-LD (continuous curve) datasets. Each step of this process involves the construction of a suffix-tree based on the training (normal) data, $S$, and the extraction of the $S$-optimal coverings for the validation (normal) and attack data. These curves have been obtained on an Intel(R) Core(TM) i7-6700HQ CPU @ 2.60GHz laptop running an Ubuntu $16.10$ release. Note that between two successive tests, the training data increases by one sequence for the UNM dataset and $100$ sequences for the ADFA-LD dataset; conversely the validation (normal) dataset is reduced by one sequence and $100$ sequences respectively. Hence, as the enrichment increases, the number of $S$-optimal coverings to extract decreases.

We observe that the elapsed time is quite stable, around $450$ seconds, for the UNM dataset and is rather independent of the enrichment step. For the ADFA-LD dataset (which is much larger in size than the UNM dataset), the elapsed time varies around an average of about $900$ seconds with a large standard deviation (about $100$ sec.). From the analysis of the logs collected during our experiments, we found that i) the time required to construct the suffix-tree is marginal, ii) the elapsed time progressively shifts from the extraction of the coverings for the normal sequences of the validation set to the extraction of the coverings for the attack data. This is due partly to the fact that the amount of normal validation data to evaluate decreases, but also because the extraction is becoming much easier for the normal data because the size of the covering becomes very small in average as the enrichment progresses (similarity scores are tangent to the maximal similarity score for normal data).

Due to the variability in length of the sequences that are added to the training data, it is difficult to progress deeper into this analysis. Nevertheless, these curves show that, i) for the various configurations corresponding to each step of the enrichment process, the extraction of the $S$-optimal coverings is bounded in time and ii) the two addressed problems are processed in a quite reasonable elapsed time. 

\section{Discussion}

The UNM dataset, although relatively old and somewhat outdated, gives a general view of the ability of the SC4ID algorithm to separate normal data from attacks sequences. Using this this dataset, we have been able to compare SC4ID with three other similarity measures whose variants are classically used in bioinformatics and text mining. From this comparative study we can conclude that SC4ID is as accurate as the Levenshtein's distance but is much faster to evaluate. SC4ID is also much more accurate and faster to evaluate than the two other similarity measures. 

We have not been able to provide a similar comparative study on the ADFA-LD dataset due to the time complexity of the other similarity measures.

The fact that the first results highlighted on the UNM dataset for the SC4ID algorithm are still observable on a much newer and reputedly difficult benchmark such as the ADFA-LD dataset is particularly promising. In both cases, the instance selection ability of SC4ID enables the rapid improvement of the separability of attack sequences from the normal ones. When the training data are sufficiently representative of the normal activity, expressed in terms of system call traces, the algorithm performs an almost perfect separation. We have shown that, if the perfect separation between normal and attack sequences is not achieved for the ADFA-LD dataset, this is because some attack sequences actually correspond exactly to some subsequences of the normal sequences that are used to train the model. If this kind of overlap is suppressed then one can expect a perfect separability as shown in Fig. \ref{fig:adfa-ls_Auc-vs-trainSize}.

In addition, the covering similarity can be used efficiently to decide whether a given normal sequence should be part of the training data or not. Basically, when a lot of normal data is available, which is actually the case for HIDS application, it offers a data selection scheme with a stop condition that can be easily set up according to the middle columns of Fig. \ref{fig:umn_distrib_roc} and \ref{fig:adfa-ld_distrib_roc}. When the ranked normal similarity curve is tangent to the 1-constant horizontal line, SC4ID will not further improve its capability to isolate anomaly data (anomaly being clearly defined in the context of the considered set of normal data). 

To give some hints about how SC4ID compares with the state of the art methods that have been tested on the ADFA-LD dataset, we report hereinafter the results we found in the recent literature:

In \cite{Xie2013}, bag of words and vector space models with \textit{tf} and \textit{tf-idf} weightings were used. Authors report a slightly less than 80\% in accuracy for a False Positive rate of 30\%.

In \cite{Xie2014}, a one-class SVM was evaluated using a feature vector composed of n-grams of length $5$. The authors report an average accuracy of $70\%$  at a False Positive Rate of about $20\%$.

In \cite{Borisaniya2015}, a $10$ fold cross-validation supervised classification (which is a much easier task than the anomaly detection task  we are considering in this study, since attack data is used to train the classifiers) has been conducted. Authors report a AUC=0.93 for the best tested method (k-nn with k=3), using an 'enhanced' vector  space model with n-grams (n=2,3,4,5).

In \cite{Kim2016}, a deep learning ensemble approach (LSTM) has been evaluated. Authors report an  AUC value of 0.928 for an aggregating method that is somewhat questionable, and an AUC value of 0.859 for a classical voting ensemble method. This implies that for a single LSTM model, the AUC value would be at most $.86$. 

In \cite{Creech2014}, a combinatorics 'semantic' approach has been set up. It requires enumerating all phrases of 5 words with gaps, each word  consisting of any subsequence (of any length) extracted from the training data (we estimate at about $6.10^16$ the number of such phrases). The authors, actually the designers of the ADFA-LD dataset, report an $AUC=.95$ after several weeks of computing effort. 

In \cite{Murtaza2015}, authors proposed to reduce system call traces by  applying trace abstraction techniques. The approach consists mainly of reducing the size of the alphabet, by using meta-symbols, i.e. subsets that partition the alphabet. Authors report a $12.69\%$ False Positive rate for a $100\%$ True Positive rate when using a HMM model.\\

All these results demonstrate that SC4ID is well positioned among leading-edge approaches.

\section{Conclusion}
In this paper, we have presented the design and implementation of SC4ID, a novel algorithm based on the concept of sequence covering, to detect abnormal sequences of system calls. It relates to a semi-supervised learning approach that it is quite efficient to detect zero-day attacks, as far as enough normal training data is available. Its main advantages compared to the state of the art approaches are: 
\begin{enumerate}[i)]
\item it is parameter-free, except for the decision threshold $\sigma$. Hence no assumption need to be made on the data, no windowing, no maximal length for the n-grams that are taken into account, no hidden architecture (HMM, LSTM) need to be defined, no meta parameters (SVM, RF) need to be tuned, etc.,
\item the $S$-optimal covering provided by the measure is interpretable and can be used to locate contextual and collective anomalies in long sequences,
\item it is incremental: the elements of the covering are progressively discovered and never modified afterwards. Hence the algorithm can be easily setup for an online exploitation,
\item it enables a fast learning curve, using the efficient sequence selection procedure that can be used to sample the (large)set of  normal sequences in order to build a minimal training set,
\item it scales well and runs easily on common hardware for medium to large size problems, thanks to its $n \cdot log(n)$ algorithmic complexity.
\end{enumerate}

However, we have pointed out a limitation of the SC4ID algorithm: it cannot separate sequences that are exact subsequences of the training set. 
However, any algorithm that uses historical data to provide a regressive or predictive scoring would also fail to correctly handle such a situation.  

As a perspective, we can expect a speed up of the algorithm when trace abstraction techniques, such as described in \cite{Murtaza2015}, are used to reduce system call traces.  This speed-up, hopefully, could potentially be obtained without losing (too much) on the accuracy. Parallelization of the algorithm is also an issue that can be addressed, in particular in the context of suffix-array implementations.

Finally, one may ask whether the covering similarity is able to offer some useful solutions in other domain such as bioinformatics, sequence mining (clustering, classification), plagiarism detection, etc. In particular, the pair-wise similarity defined in Eq. \ref{eq:SeqCoveringSimilarity} could be tested in various context, to evaluate the robustness and generality aspect of the concept of sequence covering.

%
%
%

\ifCLASSOPTIONcaptionsoff
  \newpage
\fi

\bibliographystyle{IEEEtran}
\bibliography{IEEEabrv,biblio}
%



\end{document}